\documentclass[11pt,a4paper,headinclude,footinclude]{scrreprt}

\newcommand{\app}[1]{}

\usepackage{wrapfig}
\usepackage{times}
\usepackage{boxedminipage, epsfig,graphics,graphicx, multirow, array}

\usepackage{url}
\usepackage[dvipsnames]{xcolor}

\usepackage[vlined,ruled,linesnumbered]{algorithm2e}
\SetAlFnt{\footnotesize}

\SetCommentSty{mycommfont}
\SetKwComment{tcp}{$\triangleright$\ }{}
\SetKwComment{tcc}{//\ }{}
\SetKwInput{Require}{Require}
\usepackage{listings}

\usepackage[center]{caption}

\usepackage{acronym}

\usepackage{tabularx}
\usepackage{booktabs}
\usepackage{multicol}
\usepackage{multirow}
\usepackage{graphicx}

\usepackage{subfig}
\usepackage{comment}

\usepackage{fixltx2e}

\usepackage{upgreek}
\usepackage{amsmath,amssymb,amsthm,amsfonts,amsbsy}
\usepackage{mathtools}
\usepackage{bbm}

\usepackage{thmtools, thm-restate}
\declaretheorem[name=Proposition]{proposition}

\declaretheorem[name=Definition, sibling=proposition]{definition}

\SetKwRepeat{Do}{do}{while}

\newcommand{\ALG}{{GREEDY}}
\newcommand{\alg}{{GREEDY}}

\newcommand{\ChanTrain}{{CHAN TRAIN}}

\newcommand{\OPTBTwo}{{OPT\textsubscript{B2}}}
\newcommand{\MDTOPT}{{MDTOPT}}

\newcommand{\evalFigWidth}{0.32}

\def\BibTeX{{\rm B\kern-.05em{\sc i\kern-.025em b}\kern-.08em
    T\kern-.1667em\lower.7ex\hbox{E}\kern-.125emX}}

\newtheorem{example}{Example}

\DeclarePairedDelimiter\floor{\lfloor}{\rfloor}

\newcommand{\bmax}{b_{\max}} \renewcommand{\evalFigWidth}{0.5}

\graphicspath{{figures/}}

\usepackage[utf8]{inputenc}

\usepackage{tikz}
\usepackage{lipsum}

\usepackage[headsepline,plainheadsepline,footsepline,plainfootsepline]{scrpage2}

\usepackage{mathptmx}
\usepackage[scaled=.92]{helvet}
\usepackage{courier}

\usepackage[multiuser,nomargin,marginclue,footnote]{fixme}
\fxusetheme{color}
\fxsetup{targetlayout=colorcb}

\fxsetup{status=final}

\usepackage{geometry}
\clearscrheadfoot

\ihead[\Large {\scshape TU Berlin }]{\Large {\scshape TU Berlin }}

\ifoot[{\tiny \begin{minipage}{4.0cm}Copyright at Technische Universität Berlin.\newline All Rights reserved.\end{minipage}}]{{\tiny \begin{minipage}{4.0cm}Copyright at Technische Universität Berlin.\newline All Rights reserved.\end{minipage}}}
\ofoot[Page \pagemark]{Page \pagemark}

\addtokomafont{pagefoot}{\normalfont}

\newcommand{\trnumber}{TKN-18-0003} \newcommand{\trdate}{June 2018}
\newcommand{\trauthor}{Niels Karowski, Konstantin Miller}
\newcommand{\tremail}{karowski@tkn.tu-berlin.de, konstantin.miller@ieee.org}
\newcommand{\trtitle}{Proofs and Performance Evaluation of Greedy Multi-Channel Neighbor Discovery Approaches}

\begin{document}

{
\sffamily

\thispagestyle{empty}

\setlength{\tabcolsep}{0pt} \noindent \begin{tabularx}{\columnwidth}{cXc}
  \includegraphics[height=1cm]{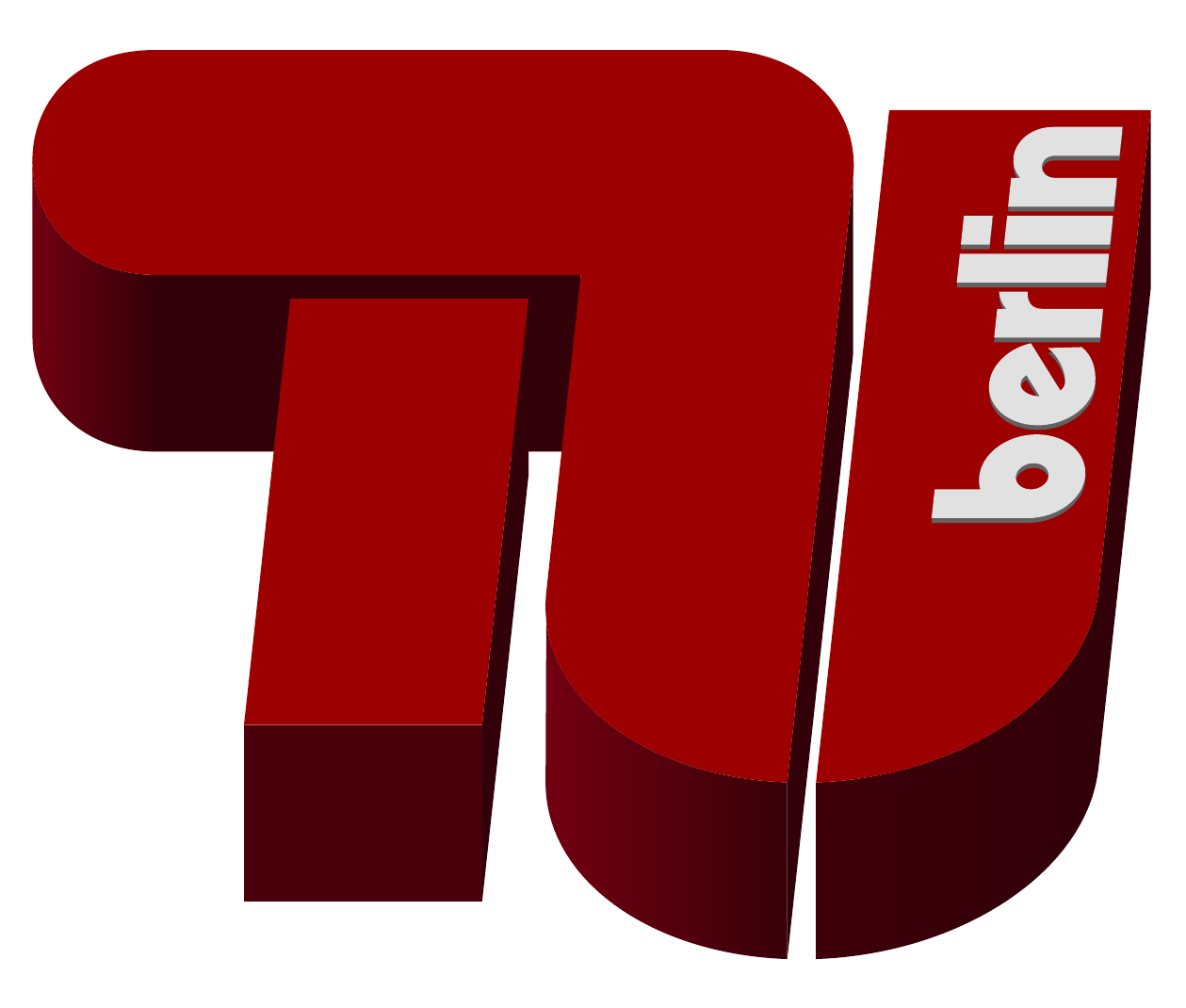}
  & &
  \includegraphics[height=1cm]{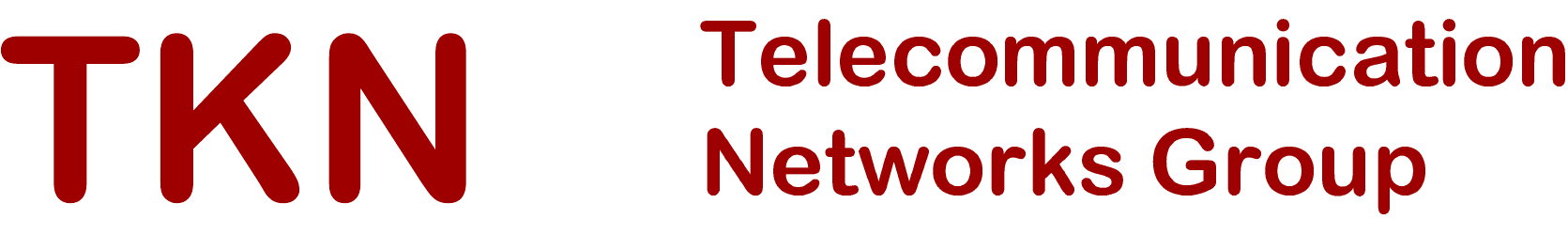}
  \\
\end{tabularx}
\setlength{\tabcolsep}{6pt} 
\vspace{1.0cm}

\begin{center}
{\huge
\noindent
Technische Universität Berlin

\vspace{0.5cm}

\noindent
Telecommunication Networks Group

\begin{center}
\rule{15.5cm}{0.4pt}
\end{center}
}
\end{center}

\begin{minipage}[][11.0cm][c]{14.5cm}
{\Huge

\begin{center}
\trtitle
\end{center}

\begin{center}
{\LARGE \trauthor} \\
{\Large \tremail}
\end{center}

\begin{center}
Berlin, \trdate
\end{center}

\vspace{0.5cm}

}

\begin{center}
\setlength{\fboxrule}{2pt}\setlength{\fboxsep}{2mm}
\fbox{TKN Technical Report \trnumber}
\end{center}

\end{minipage}

\setlength{\fboxrule}{0.4pt}
\setlength{\fboxsep}{0.4pt}

\begin{center}

  \rule{15.5cm}{0.4pt}

  \vspace{0.5cm}

  {\huge {TKN Technical Reports Series}}

  \vspace{0.5cm}

  {\huge Editor: Prof. Dr.-Ing. Adam Wolisz}

  \vspace{0.5cm}

 \end{center}

}
 
\begin{abstract}
\subsection*{\abstractname}
The accelerating penetration of physical environments by objects with information processing and wireless communication capabilities requires approaches to find potential communication partners and discover services. 
In the present work, we focus on passive discovery approaches in multi-channel wireless networks based on overhearing periodic beacon transmissions of neighboring devices which are otherwise agnostic to the discovery process. 
We propose a family of low-complexity algorithms that generate listening schedules guaranteed to discover all neighbors. The presented approaches 
simultaneously depending on the beacon periods optimize the worst case discovery time, the mean discovery time, and the mean number of neighbors discovered until any arbitrary in time. 
The presented algorithms are fully compatible with technologies such as IEEE~802.11 and IEEE~802.15.4. 
Complementing the proposed low-complexity algorithms, we formulate the problem of computing discovery schedules that minimize the mean discovery time for arbitrary beacon periods as an integer linear problem. We study the performance of the proposed approaches analytically, by means of numerical experiments, and by extensively simulating them under realistic conditions. We observe that the generated listening schedules significantly -- by up to factor 4 for the mean discovery time, and by up to 300\% for the mean number of neighbors discovered until each point in time -- outperform the Passive Scan, a discovery approach defined in the IEEE 802.15.4 standard. Based on the gained insights, we discuss how the selection of the beacon periods influences the efficiency of the discovery process, and provide recommendations for the design of systems and protocols.

\end{abstract}

\section{Introduction}

We are currently observing a rapid augmentation of physical objects surrounding us with information processing and wireless communication capabilities. It is estimated that by 2020 25~\cite{gartnerIoT} up to 50~\cite{cisco_IoT} billion objects will be connected to the Internet. This development is leading us to a new era of computing. The resulting network of ``smart'' objects that interact with each other and exchange information without a direct human intervention, the so-called \ac{IoT}, will serve as a foundation for novel applications in a wide range of domains. 

In order to discover services of interest devices will need to detect other entities within communication range that are able to use common communication technology---the so called neighbors.

Neighbor discovery can be done in two fundamentally different ways. For an \emph{active} discovery, the discoverer broadcasts probe requests that must be answered by the neighbors.  
An active discovery is fast but has the drawback that all neighbors have to consume energy by (continuously) listening to potential inquiries even though they might only be interested in being detected but not in discovering their own neighborhood.
In contrast, \emph{passive} schemes perform the discovery by overhearing beaconing messages that are periodically broadcasted (with a specific \ac{BP}) by neighbors interested in being discovered.
The beaconing neighbors themselves are hereby agnostic to the discovery process. 
Let us emphasize that periodic beaconing is already used in many widely deployed technologies such as IEEE~802.11~\cite{ieee80211} and IEEE~802.15.4~\cite{ieee802154}.
In order to be compatible with current state-of-the-art technologies, such as IEEE~802.11 and IEEE~802.15.4, neighbor discovery must support multi-channel environments. 
Finally, we assume lack of time synchronization among the devices involved in the discovery process.

A frequently adopted objective for the design of discovery approaches is the minimization of the \ac{WDT}---time required to detect all \emph{potential} neighbors. A complete discovery is desirable, e.g., in order to avoid interference with neighbors when establishing a new network. In addition, minimizing the \ac{WDT} has the advantage of implicitly minimizing the consumed energy. 
Other applications are interested in the maximization of the number of discoveries until a given point in time, which we call the \ac{NDoT}, as, e.g., in the case of identifying potential forwarders in \aclp{DTN}.
Yet other applications benefit from discovering the individual neighbors as early as possible, e.g., emergency services. Their objective is thus the minimization of the \ac{MDT}. 
Since many devices in IoT environments will be battery powered, and will have limited computational resources, neighbor discovery should be performed in an energy-efficient way with low to moderate computational requirements.

The publication~\cite{Karowski18} associated with this report  presents several novel contributions providing simple and efficient discovery algorithms applicable under realistic conditions:

We provide for the first time a full characterization of the class of listening schedules that are guaranteed to discover all neighbors (we call such schedules \emph{\textbf{complete}}), and that pointwise maximize the \ac{CDF} of the discovery times. The latter feature implies that they optimize all three mentioned performance metrics \emph{simultaneously}: \ac{WDT}, \ac{MDT}, and the \ac{NDoT}. We call these schedules \emph{\textbf{recursive}}, due to their specific structure.

Our second, practically most relevant, contribution consists of several approaches to construct listening schedules that, under certain assumptions, are recursive (and thus inherit the corresponding optimality properties). We define a family of low-complexity algorithms that we call \alg{}, due to their operation mode~\cite{West2001}. Further, we define an algorithm called CHAN TRAIN which is an extension of the \alg{} family, aiming at a reduction of the number of channel switches. 

In general, the performance of discovery algorithms is strongly dependent on the allowed set of \acp{BP}, 
i.e. periods with which beacons are transmitted. To this point - incompatible with the state-of-the-art wireless protocols -
assumptions about the beacon transmission patterns have been made so far in the literature. In this paper, we consider on one hand the most general case $\mathbb{F}_1$, a family containing all
possible \ac{BP} sets, but introduce in addition also two other practically important families of \ac{BP} sets---$\mathbb{F}_2$ and $\mathbb{F}_3$. They include all \ac{BP} sets supported by IEEE 802.15.4 and a large part
of the BI sets supported by IEEE 802.11---two widely adopted standards for wireless communication. 

We prove that for \ac{BP} sets from $\mathbb{F}_3$ the listening schedules computed by GREEDY and CHAN TRAIN are recursive (and thus complete, and optimal w.r.t. the three targeted performance metrics). Moreover, for \ac{BP} sets from $\mathbb{F}_2$ the computed schedules are complete and \ac{WDT}-optimal, while they are close-to-optimal w.r.t. the \ac{MDT}. Finally, we show that even for the most general case of $\mathbb{F}_1$ the computed schedule are complete, close-to-optimal w.r.t. the \ac{MDT}, and still within 30\%  of the optimum for the \ac{WDT}, while this gap decreases for an increasing number of channels.

Our third contribution demonstrates that even for arbitrary \ac{BP} sets from $\mathbb{F}_1$
complete and \ac{MDT}-optimal schedules are achievable, albeit only by solving an \ac{ILP}.
We prove that computed schedules are also \ac{WDT}-optimal for \ac{BP} sets from $\mathbb{F}_2$ and \ac{NDoT}-optimal for \ac{BP} sets from $\mathbb{F}_3$. This approach is attractive due to the broad range of supported \ac{BP} sets. However, it has a high computational complexity and memory consumption, restricting its usage to offline computations, and to scenarios with a moderate number of channels and size of the used \acp{BP}.

As additional contribution, we define an algorithm called \OPTBTwo{} that computes recursive schedules for scenarios, in which the cardinality of the \ac{BP} set is restricted to two entries. 

Unfortunately performing of such discovery in real multi-channel environments suffers under an implementation impact: non-negligible deaf periods occur during the execution of a channel switch resulting in potentially missing some beacons transmitted during such deaf periods. 
Due to this effect even algorithms provably generating complete schedules will, in reality, miss some neighbors - the percentage of missed neighbors can reasonably be expected to increase with the increase of the number of channel switches required by a given algorithm. 
In order to quantify this impact, we perform simulations using a realistic wireless model and device behavior expressing the results in the form of an additional performance metric --- the success rate, which is the fraction of neighbors discovered under this realistic conditions by any algorithm under consideration.
Using this additional performance metric we have derived our next contribution: We suggest two instances of the \ALG{} family of algorithms designed to reduce the number of channel switches and perform their simulative performance evaluation w.r.t achievable fraction of discovered neighbors.

In all evaluations, in addition to comparing the performance with the optimum, we perform a comparison against \acf{PSV}, a discovery scheme defined by the IEEE 802.15.4 standard. We observe that \ALG{} algorithms significantly (by up to several hundreds percent) outperform \ac{PSV} w.r.t. the \ac{MDT} and the \ac{NDoT} in all studied scenarios.

As our final contribution we discuss the strong impact the structure of allowed \ac{BP} sets has on the performance of discovery approaches, and provide recommendations for a \ac{BP} selection that supports efficient neighbor discovery. These recommendations may be useful, on the one hand, for the development of novel wireless communication based technologies that use periodic beaconing messages for management or synchronization purposes, and, on the other hand, for the \ac{BP} selection for existing technologies that support a wide range of \acp{BP}, such as IEEE~802.11.

This report is a supplement to~\cite{Karowski18} and is structured as follows. 
In Section~\ref{sec:examples} additional examples of the properties of the  developed discovery strategies are provided.
Section~\ref{sec:opt_cdt} provides a proof on the optimal duration of a complete discovery.
Properties of recursive schedules regarding the optimality relating to the analyzed performance metrics are described in Section~\ref{app:recursive}.
A performance analysis for the \ALG{} algorithms is provided in Section~\ref{sec:optim_greedy}, for the \ChanTrain{} algorithm in Section~\ref{sec:optim_chantrain} and for the \OPTBTwo{} algorithm in Section~\ref{sec:optim_optbtwo}.
Section~\ref{app:mdt_opt_cdt_opt} provides a proof on the \ac{WDT} of \ac{MDT}-optimal schedules.
The relation between schedules being \ac{NDoT}-optimal and \ac{WDT}-optimal as well as \ac{MDT}-optimal is discussed in Section~\ref{app:mdt_opt_ndot_opt}. 
The computational complexity and memory requirements of a possible \ALG{} implementation is presented in Section~\ref{sec:complexity}.
Finally, Section~\ref{sec:appendix_results} provides additional simulation results.

\section{Examples}
\label{sec:examples}

In the following figures gray squares represent the scanned channels, while the numbers in each square represent the expected value of the fraction of neighbors that can be discovered by scanning a particular channel during a particular time slot. In addition to the assumptions on the uniform distribution of channels and beacon offsets described in Section~III in~\cite{Karowski18}, in the following examples we also assume a uniform distribution of \acp{BP}. That is, the configuration probabilities are assumed to be $P_\kappa=\frac{1}{b_\kappa|B||C|}$, for each $\kappa\in K_{BC}$. Finally, we remark that while the presented examples consider very small scenarios, the illustrated performance gaps may become arbitrarily large if larger scenarios are considered.

\begin{example}[A recursive schedule does not always exist]
\label{exa:recursive}
Consider the setting depicted in Figure~\ref{fig:NonStructuredSchedule}, with $B = \{1,2,3\} \notin \mathbb{F}_2$ and $|C| = 2$.  Vertical dashed lines mark time instants until which all configurations $\left(c,b_i,\delta\right)$ for a \ac{BP} $b_i$ have to be discovered on each channel $c \in C$ in order for the schedule to be recursive. Consequently, gray boxes indicate channels that have to be scanned during the first 4 time slots (uniquely determined up to swapping the channels). Observe that a recursive listening schedule for this example does not exist due to the fact that it would have to discover all remaining neighbor configurations with $b=3$ during the two time slots 5 and 6. However, the number of remaining configurations is 3: $(0, 3, 0)$, $(1, 3, 1)$, and $(1, 3, 2)$.
\end{example}

\begin{figure}[h]
\centering
\includegraphics[width=0.40\textwidth]{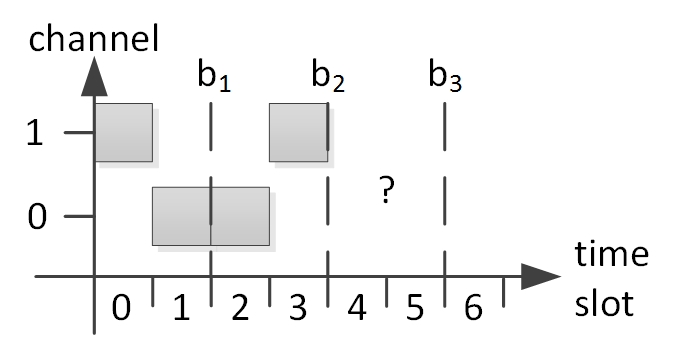}
\caption{Illustration for Example~\ref{exa:recursive}. With $B = \{1,2,3\} \notin \mathbb{F}_2$ and $|C| = 2$, a recursive schedule does not exist.}
\label{fig:NonStructuredSchedule}
\end{figure}

\begin{example}[\ALG{} is not optimal for arbitrary \ac{BP} sets]
\label{exa:greedy_suboptimal}
For arbitrary \ac{BP} sets, a \ALG{} algorithm does not necessarily generate \ac{MDT}-optimal or \ac{WDT}-optimal schedules. Consider a setting with $B = \{1,2,3\}\in\mathbb{F}_1\setminus\mathbb{F}_2$ and $|C| = 3$. Figure~\ref{fig:GREEDY_F1_non_makespan_MDT} shows a \ALG{} listening schedule, which is neither \ac{WDT}-optimal nor \ac{MDT}-optimal, while Figure~\ref{fig:GREEDY_F1_non_makespan_MDT_GENOPT_comp} depicts a schedule, which is both \ac{WDT}-optimal and \ac{MDT}-optimal but not GREEDY.
In particular, optimal value for the \ac{WDT} in this example constitutes 9 time slots, while optimal value for \ac{MDT} is $2.6\overline{1}$  time slots, in contrast to $2.7\overline{2}$ time slots achieved by the schedule computed by \ALG{} . Another example using $B = \{2,3,4,6,12\} \in \mathbb{F}_2 \setminus \mathbb{F}_3$ and $|C| = 2$ is depicted in Figure~\ref{fig:GREEDY_non_makespan_MDT_opt_F2}. The listening schedule generated by \ALG{} is \ac{WDT}-optimal but only achieves a \ac{MDT} of $5.3$ time slots in comparison with the optimal value of $5.1$ time slots.
\end{example}

\begin{figure}[h]
\centering
		\subfloat[GREEDY]{		
        \includegraphics[width=0.7\textwidth]{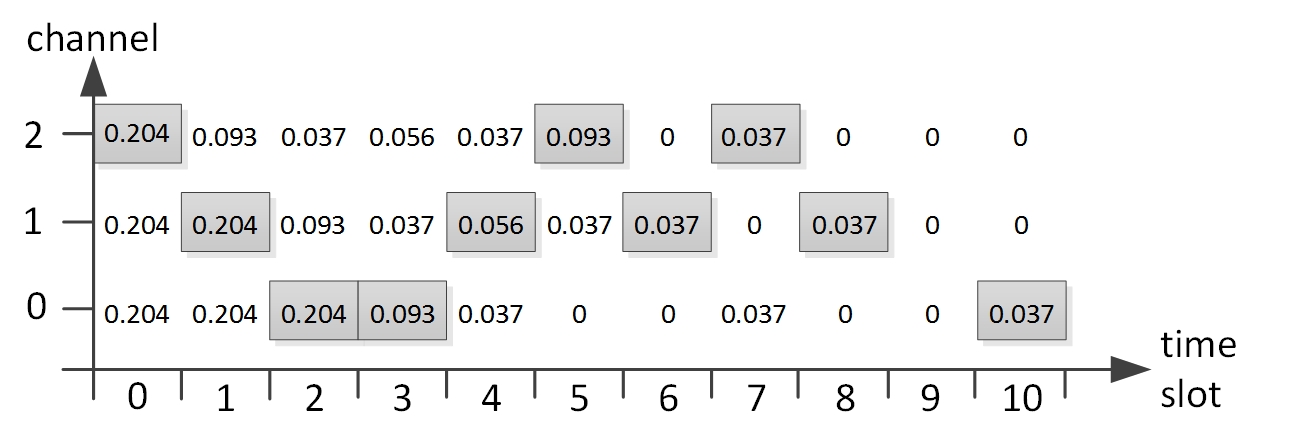}
				\label{fig:GREEDY_F1_non_makespan_MDT}
		} \\
		\subfloat[\ac{MDT}-optimal schedule]{
        \includegraphics[width=0.7\textwidth]{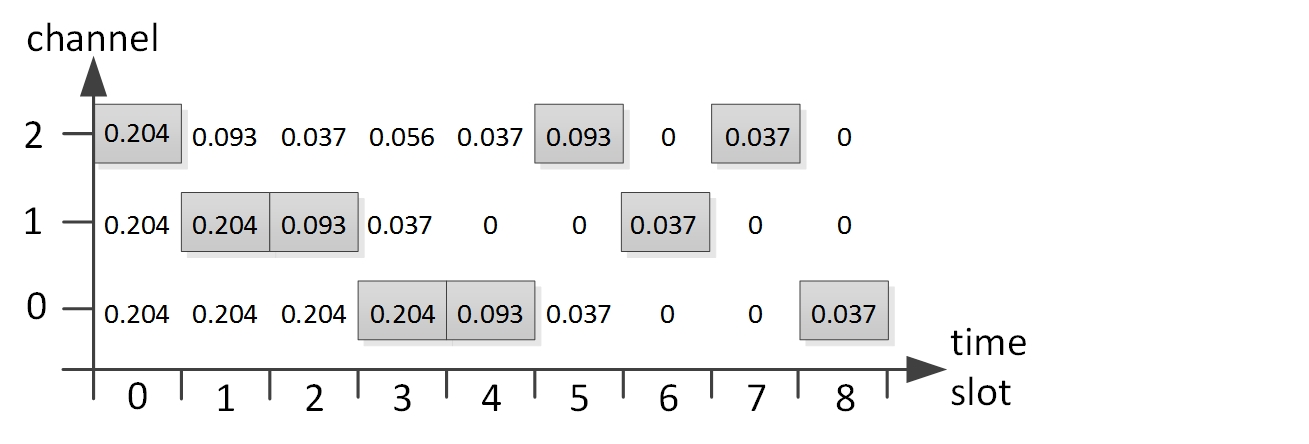}
				\label{fig:GREEDY_F1_non_makespan_MDT_GENOPT_comp}
    }
\caption{Illustration for Example~\ref{exa:greedy_suboptimal} showing a schedule generated by \ALG{} for $B = \{1,2,3\} \not\in \mathbb{F}_2$ and $|C| = 3$ that is neither \ac{MDT}-optimal nor \ac{WDT}-optimal.}
\label{fig:GREEDY_non_makespan_MDT_opt_F1}
\end{figure}

\begin{example}[CHAN TRAIN is not \ALG{} for general \ac{BP} sets]
\label{exa:chantrain_nongreedy}
For a \ac{BP} set $B\not\in\mathbb{F}_3$ Proposition~\ref{prop:chantraingreedy} is no longer true. Figure~\ref{fig:CHAINTRAIN_nonGreedyBehaviour} shows an example setting with $B = \{1,2,3,6\} \in \mathbb{F}_2\setminus\mathbb{F}_3$ and $|C| = 3$, in which this is the case.
\end{example}

\begin{example}[\MDTOPT{} is in general not \ac{WDT}-optimal]
\label{exa:genopt_notmakespanoptimal}

Listening schedules generated by \MDTOPT{} might not be \ac{WDT}-optimal. Figure~\ref{fig:GENOPT_F1_non_makespan_opt} depicts an example using $B = \{1,2,4,5\}$ and $|C| = 2$, for which a \ac{MDT}-optimal listening schedule cannot be constructed within $\max(B) |C|$ time slots. Figure~\ref{fig:GENOPT_F1_makespan_tmax-optMakespan} shows a \ac{MDT}-optimal listening schedule with the additional constraint of using at most $\max(B)|C|$ time slots. Observe that its \ac{MDT} is $2.875$  time slots as compared to the optimal value of $2.75$ time slots of the schedule shown in Figure~\ref{fig:GENOPT_F1_makespan_tmax-lcm}. 
\end{example}

\begin{figure*}
\centering
			\subfloat[GREEDY]{		
        \includegraphics[width=1.0\textwidth]{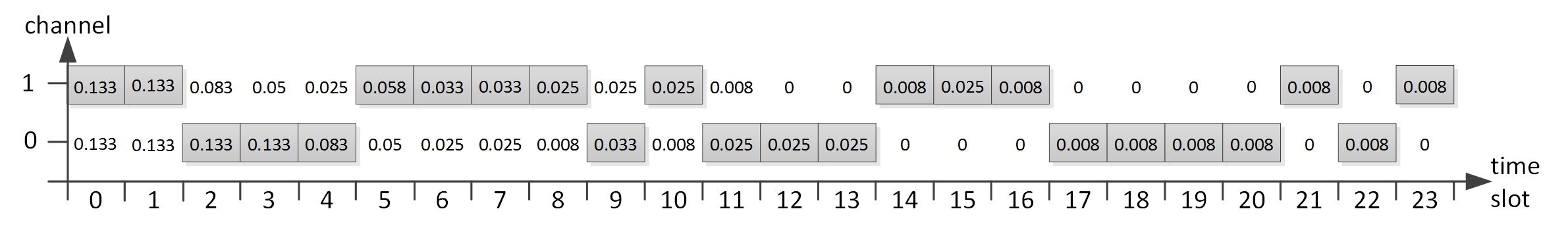}
				\label{fig:GREEDY_F2_non_MDT}
		} \\
		\subfloat[\ac{MDT}-optimal schedule]{
        \includegraphics[width=1.0\textwidth]{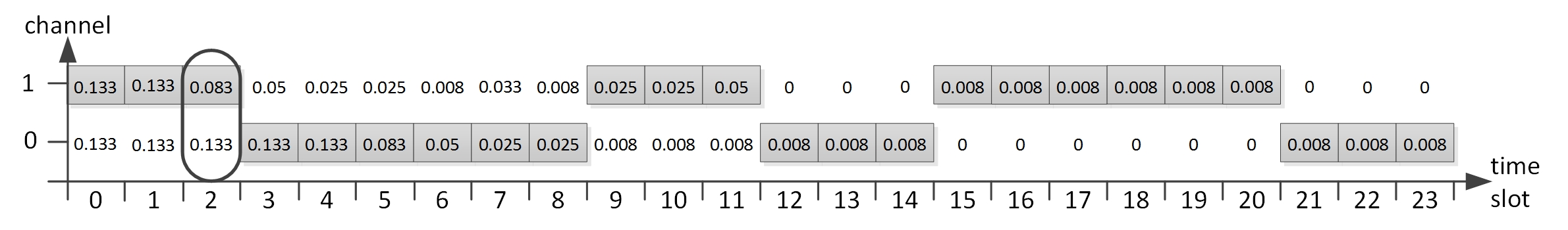}
				\label{fig:GREEDY_F2_non_MDT_GENOPT}
    }
\caption{Illustration for Example~\ref{exa:greedy_suboptimal} showing a schedule generated by \ALG{} for $B = \{2,3,4,6,12\} \in \mathbb{F}_2 \setminus \mathbb{F}_3$ and $|C| = 2$ that is \ac{WDT}-optimal but not \ac{MDT}-optimal.}
\label{fig:GREEDY_non_makespan_MDT_opt_F2}
\end{figure*}

\begin{figure*}
\centering
		\subfloat[CHAN TRAIN]{		
        \includegraphics[width=0.7\textwidth]{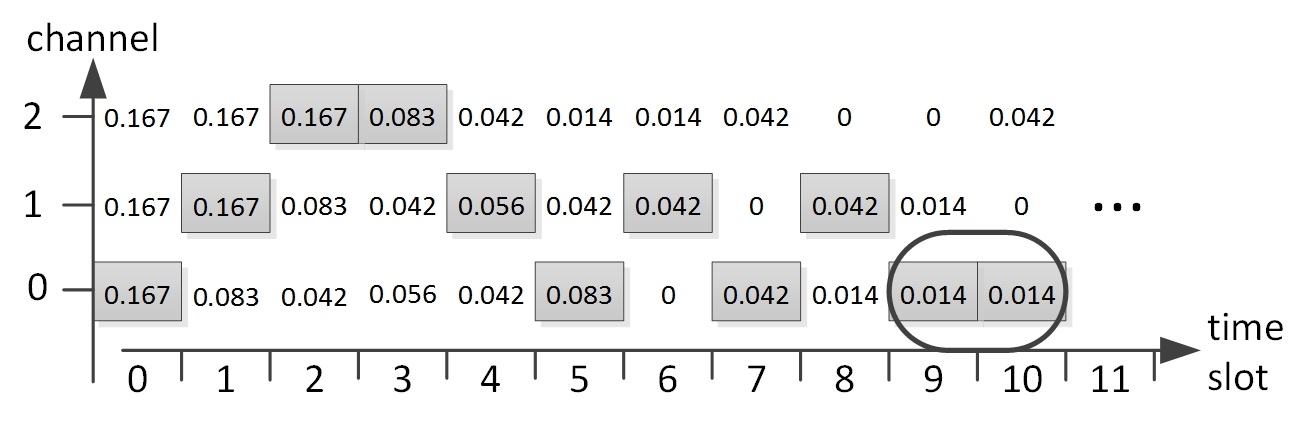}
				\label{fig:CHANTRAIN_F2}
		}\\
		\subfloat[GREEDY]{
        \includegraphics[width=0.7\textwidth]{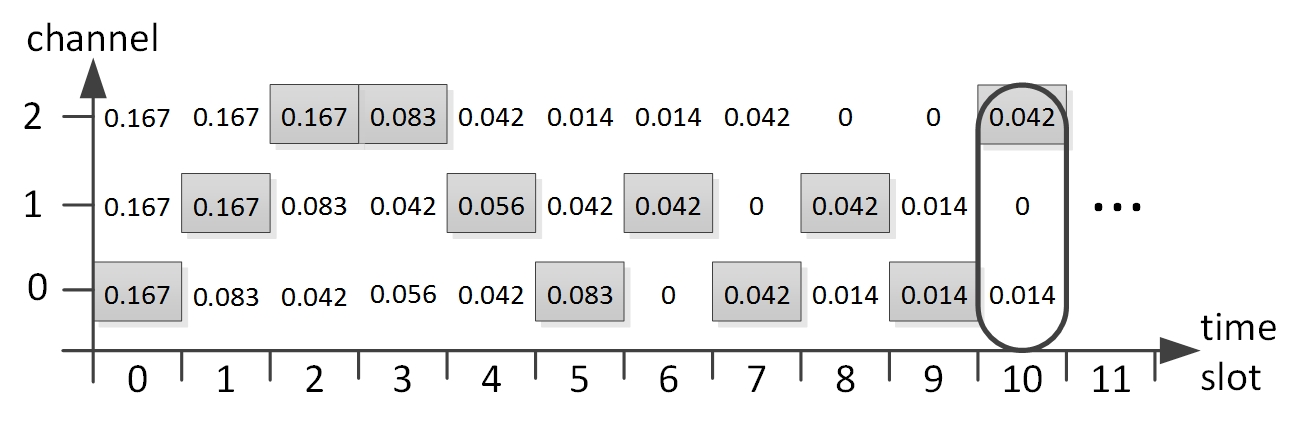}
				\label{fig:CHANTRAIN_GREEDY_F2}
    }
\caption{Illustration for Example~\ref{exa:chantrain_nongreedy} showing the non-greedy behavior of the CHAN TRAIN strategy for the BI set $B = \{1,2,3,6\} \in \mathbb{F}_2\setminus\mathbb{F}_3$ and $|C| = 3$.}
\label{fig:CHAINTRAIN_nonGreedyBehaviour}
\end{figure*}

\begin{figure}
\centering
		\subfloat[\ac{WDT} constrained by the optimum value $\max(B)  |C|$]{		
        \includegraphics[width=0.5\textwidth]{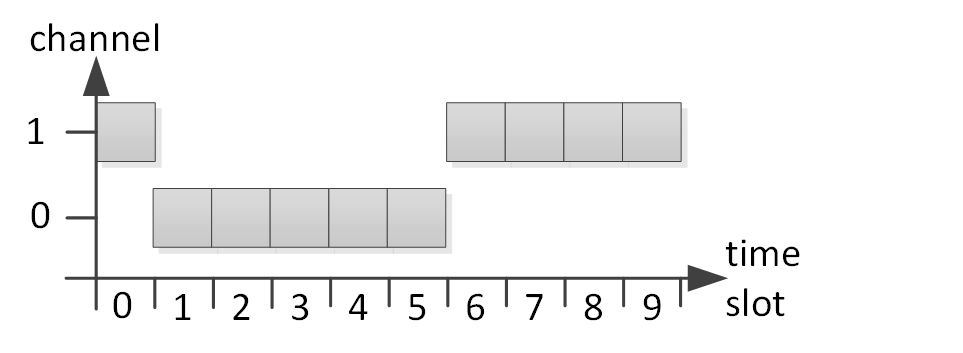}
				\label{fig:GENOPT_F1_makespan_tmax-optMakespan}
		} \\
		\subfloat[\ac{WDT} constrained by $LCM(B)  |C|$]{
        \includegraphics[width=0.5\textwidth]{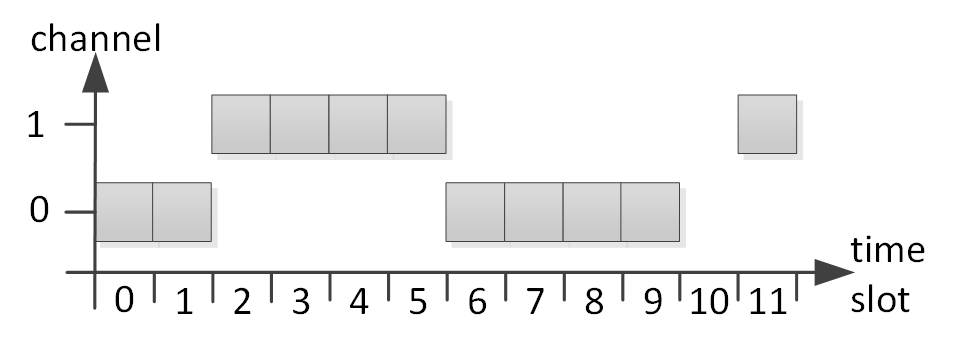}
				\label{fig:GENOPT_F1_makespan_tmax-lcm}
    }
\caption{Illustration for Example~\ref{exa:genopt_notmakespanoptimal}, with \ac{BP} set $B = \{1,2,4,5\} \in \mathbb{F}_1 \setminus \mathbb{F}_2$ and $|C| = 2$. In this example, no schedule exists which is both \ac{MDT}-optimal and \ac{WDT}-optimal.}
\label{fig:GENOPT_F1_non_makespan_opt}
\end{figure}

\section{Optimal \ac{WDT}}
\label{sec:opt_cdt}

In this section, we prove that the optimum duration of a complete discovery is $\max(B)|C|$ time slots for arbitrary \ac{BP} sets and channel sets. This allows us to define a \ac{WDT}-optimal listening schedule as a schedule which requires $\max(B)|C|$ for a complete discovery.

\begin{proposition}[\ac{WDT}-optimality]
\label{prop:cdt_optimality}
For an arbitrary set of \acp{BP} $B$, and a set of channels $C$, the optimum duration of a complete discovery is $\max(B)\left|C\right|$ time slots.
\end{proposition}
\begin{proof}
Observe that the complete discovery can always be performed within $\max(B)|C|$ time slots by listening on each channel for $\max(B)$ consecutive time slots. More precisely, channel $c_j$, $j\in\{0,\ldots,|C|-1\}$ is scanned during the time slots $\{j\bmax,\ldots,(j+1)\bmax-1\}$, where $\bmax=\max(B)$. Thus, configuration $(c_j,b,\delta)$, $b\in B$, $\delta\in\{0,\ldots,b-1\}$, is discovered during the time slot $t=j\bmax+\delta$. Consequently, all configurations from $K_{BC}$ are discovered during the time slots $\{0,\ldots,\bmax|C|-1\}$.

To see that a complete discovery cannot be performed faster than $\max(B)|C|$, consider the configurations $\left\{\,(c,\bmax,\bmax-1)\mid c\in C\,\right\}\subseteq K_{BC}$. To discover them, it is necessary to scan at least the time slots $\left\{i\bmax+\bmax-1\right\}_{i\in\{0,\ldots,|C|-1\}}$, proving the claim.
\end{proof}

Consequently, we define a \ac{WDT}-optimal schedule as follows.

\begin{definition}[WDT-optimal schedule]
We call a listening schedule for a \ac{BP} set $B$ and a channel set $C$ \ac{WDT}-optimal if it discovers all neighbors within $\max(B)|C|$ time slots.
\end{definition}

Note that \ac{WDT}-optimal schedules do not have idle time slots as stated in the following proposition.

\begin{proposition}[WDT-optimality implies no idle slots]
\ac{WDT}-optimal schedules have no idle time slots.
\end{proposition}
\begin{proof}
Assume that a time slot $t\in\{0,\ldots,\bmax|C|-1\}$, where $\bmax=\max(B)$, $B$ is the set of \acp{BP}, and $C$ is the set of channels, is idle. Then, since the schedule is \ac{WDT}-optimal, there remain $|C|-1$ time slots $\left\{\,i\bmax+t\bmod\bmax\mid i\in\{0,\ldots,|C|-1\}\right\}\setminus\{t\}$, during which a configuration from the set $\left\{\,(c,\bmax,t\bmod \bmax)\mid c\in C\,\right\}$ can be discovered. Since the latter set of configurations has size $|C|$, we obtain a contradiction.
\end{proof}

\section{Recursive Schedules}
\label{app:recursive}

Recursive schedules have a number of compelling properties that are stated in the following proposition.

\begin{proposition}[Properties of recursive schedules]
\label{prop:recursive_optimal}
A recursive schedule for a \ac{BP} set $B\subset\mathbb{N}^+$ is complete, \ac{WDT}-optimal, \ac{MDT}-optimal, and \ac{NDoT}-optimal.
\end{proposition}
\begin{proof}
Note that recursive schedules are complete and \ac{WDT}-optimal by definition (see Definition 1 in~\cite{Karowski18}).

Further, note that a schedule $\mathcal{L}$ that maximizes the \ac{NDoT} is also \ac{MDT}-optimal since it pointwise maximizes the \ac{CDF} of discovery times
\begin{align*}
\sum_{\kappa \in K_{BC}: \mathcal{T}_\kappa(\mathcal{L}) \leq t} P_{\kappa} \geq \sum_{\kappa \in K_{BC}: \mathcal{T}_\kappa(\mathcal{L'}) \leq t} P_{\kappa} \\
\text{ for all } \mathcal{L'}\subset C\times\mathbb{N}, t \in \mathbb{N} \,.
\end{align*}
Consequently, it is sufficient to show \ac{NDoT}-optimality. To prove it, we show that for a recursive schedule $\mathcal{L}$, each scan $(c,t)\in\mathcal{L}$ with $t\in\{0,\ldots,b|C|-1\}$ results in the discovery of the configuration $(c,b,t\bmod b)$. This implies that for each $b\in B$, the number of configurations detected until each time slot is maximized since it is not possible to discover more than one configuration with a certain \ac{BP} $b\in B$ per time slot. Due to the assumption that the channels and the beacon offsets for each \ac{BP} are uniformly distributed, it also implies that the sum of the discovery probabilities is maximized in each time slot, proving the claim.

Observe that for a $t\in\{0,\ldots,b|C|-1\}$, there are $|C|$ configurations with a \ac{BP} $b$ and an offset $t\bmod b$, namely the configurations $\left\{\,(c,b,t\bmod b)\mid c\in C\,\right\}$. At the same time, within the first $b|C|$ time slots there are exactly $|C|$ time slots where these configurations can be discovered, namely the time slots $\left\{\,ib+t\bmod b\mid i\in\{0,\ldots,|C|-1\}\,\right\}$. Consequently, since a recursive schedule per definition detects all configurations with the \ac{BP} $b$ during the first $b|C|$ time slots, it has to scan a different channel during each of the time slots $\left\{\,ib+t\bmod b\mid i\in\{0,\ldots,|C|-1\}\,\right\}$. It is thus detecting one configuration with the \ac{BP} $b$ and offset $t\bmod b$ in each of these time slots.
\end{proof}

Unfortunately, not each scenario admits a recursive schedule, see Example~\ref{exa:recursive} in Section~\ref{sec:examples}. However, for \ac{BP} sets from $\mathbb{F}_3$ recursive schedules always exist, as shown by the following proposition.

\begin{proposition}
\label{prop:recursive_F3}
For a \ac{BP} set $B \in \mathbb{F}_3$ a recursive schedule always exists.
\end{proposition}
\begin{proof}
We prove the claim by induction. First, we show that a recursive schedule exists for a \ac{BP} set containing just one element. Then, we show that a recursive schedule for a \ac{BP} set $B'\in\mathbb{F}_3$ can be extended to a recursive schedule for a \ac{BP} set $B''\in\mathbb{F}_3$, with $B'\subset B''$, and $|B''|=|B'|+1$.

To see that a recursive schedule for a \ac{BP} set $B=\{b\}$ exists, consider the schedule that scans channel $c_j$, $j\in\{0,\ldots,|C|-1\}$, during the time slots $\{jb,\ldots,(j+1)b-1\}$. This schedule discovers all configurations during the time slots $\{0,\ldots,|C|b-1\}$, and is thus recursive.

Now, assume that we are given a recursive schedule for a \ac{BP} set $B'$. We show that it can be extended to a recursive schedule for $B''$. We denote $\bmax'=\max(B')$ and $\bmax''=\max(B'')$. W.l.o.g. we assume $b''>b'$. Observe that $B''\in\mathbb{F}_3$ implies $b''=\alpha b'$, for an $\alpha\in\mathbb{N}_+$. 

Observe that any complete and \ac{WDT}-optimal schedule for a \ac{BP} set $B$ can be written in form of a matrix $A\in\mathbb{N}^{|C|\times b_{\max}}$, $\bmax=\max(B)$, where each element $A_{c\delta}$ indicates that channel $c$ is scanned during the time slot $A_{c\delta}b_{\max}+\delta$, and thus the configuration $(c,\bmax,\delta)$ is discovered during this time slot. For a matrix $A$ to be a valid representation of a complete schedule, it must hold $A_{c\delta}\neq A_{c'\delta}$, for $c\neq c'$, $\delta\in\{0,\ldots,\bmax-1\}$, since otherwise two different channels would have to be scanned during the same time slot. Furthermore, it must hold $A_{c\delta}\in\{0,\ldots,|C|-1\}$ since a \ac{WDT}-optimal schedule only uses time slots $t\in\{0,\ldots,b_{\max}|C|-1\}$.

Assume that the given recursive schedule for the \ac{BP} set $B'$ is represented by the matrix $A'\in\mathbb{N}^{|C|\times b'_{\max}}$. We need to show that $A'$ can be transformed to a new matrix $A''\in\mathbb{N}^{|C|\times b''_{\max}}$, such that $A''$ is a valid representation of a complete and \ac{WDT}-optimal schedule, and such that the channels scanned during the time slots $t\in\{0,\ldots,\bmax'|C|-1\}$ by the schedule $A'$ are also scanned during the same time slots by the schedule $A''$.

We claim that such a transformation is obtained by mapping each element $A'_{c\delta'}$ to an element $A''_{c\delta''}$ as follows:
\begin{align*}
\begin{cases}
\delta''&=(A'_{c\delta'}b'_{\max}+\delta')\bmod b''_{\max}\\
A''_{c\delta''}&=\frac{(A'_{c\delta'}b'_{\max}+\delta')-\delta''}{b''_{\max}}
\end{cases}\,.
\end{align*}
Observe that since the numerator is an integer multiple of the denominator, the expression on the right-hand side evaluates to an integer.

This mapping retains the channels scanned by $A'$ since
\begin{align*}
A''_{c\delta''}b''_{\max}+\delta''&=(A'_{c\delta'}b'_{\max}+\delta')-\delta''+\delta''\\
&=A'_{c\delta'}b'_{\max}+\delta'\,.
\end{align*}

Moreover, this mapping is injective, such that no two elements from $A'$ are mapped to the same element of $A''$. This is due to the fact that $\delta''=(A'_{c\delta'}b'_{\max}+\delta')\bmod b''_{\max}$ is an injective function of $\delta'$, since $\bmax''=\alpha\bmax'$, $\alpha\in\mathbb{N}_+\setminus\{1\}$.

To show that $A''_{c\delta''}<|C|$, we need to show that $A'_{c\delta'}b'_{\max}+\delta'<|C|b''_{\max}$. 
Observe that
\begin{align*}
A'_{c\delta'}b'_{\max}+\delta'&<|C|b'_{\max}+\delta'\\
&<|C|b'_{\max}+b'_{\max}=(|C|+1)b'_{\max}\\
&\leq|C|\alpha b'_{\max}\\
&<|C|b''_{\max}\,.
\end{align*}

Finally, we have to show that $A''_{c\delta''}\neq A''_{c'\delta''}$ for $c\neq c'$. To see that, observe that if two elements from the matrix $A'$ are transformed onto the same column in the matrix $A''$, they must have been in the same column already in $A'$, since, as already mentioned earlier, $\delta''=(A'_{c\delta'}b'_{\max}+\delta')\bmod b''_{\max}$ is an injective function of $\delta'$. Consequently, for their preimages must hold $A'_{c\delta'}\neq A'_{c'\delta'}$. But then we obtain
\begin{align*}
A'_{c\delta'}&\neq A'_{c'\delta'}\\
\Leftrightarrow\quad\frac{(A'_{c\delta'}b'_{\max}+\delta')-\delta''}{b''_{\max}}&\neq\frac{(A'_{c'\delta'}b'_{\max}+\delta')-\delta''}{b''_{\max}}\\
A''_{c\delta''}&\neq A''_{c'\delta''}\,.
\end{align*}

Finally, the remaining empty rows in each column of $A''$ can be filled arbitrarily, but such that for the added elements the two properties of a matrix representation still hold: $A''_{c\delta''}\neq A''_{c'\delta''}$ for $c\neq c'$, and $A''_{c\delta''}\in\{0,\ldots,|C|-1\}$.
\end{proof}

\section{Performance Analysis for \ALG{} Algorithms}
\label{sec:optim_greedy}

In this section, we formulate and prove two results on the optimality of the \ALG{} algorithms. The first claim that we prove is that \ALG{} algorithms generate listening schedules which are \ac{WDT}-optimal (and consequently also complete) for \ac{BP} sets from $\mathbb{F}_2$. This result is stated in the following proposition.

\begin{proposition}
\label{prop:greedy_f2}
For a \ac{BP} set $B\in\mathbb{F}_2$, schedules generated by \ALG{} are \ac{WDT}-optimal.
\end{proposition}
\begin{proof}
The idea for the proof is as follows. We first show that scanning a channel $c\in C$ in a time slot $t\in\{0,\ldots,\bmax|C|-1\}$ either results in the discovery of at least the configuration $(c,\bmax,t \bmod \bmax)$, or it results in the discovery of 0 configurations. We then show that this implies that a \ALG{} algorithm discovers all configurations $\left\{\,(c,\bmax,\delta)\mid c\in C,\,\delta\in\{0,\ldots,\bmax-1\}\,\right\}$ during the first $\bmax|C|$ time slots. Finally, we show that in $\mathbb{F}_2$, this implies that all configurations are discovered.

Assume that scanning a channel $c\in C$ in a time slot $t\in\{0,\ldots,\bmax|C|-1\}$ does not result in the discovery of the configuration $(c,\bmax,t \bmod \bmax)$. This implies that channel $c$ has already been scanned during a time slot $t'\in\left\{i\bmax+t\bmod\bmax\right\}_{i\in\{0,\ldots,\floor{t/\bmax}-1\}}$, so that the configuration $(c,\bmax,t \bmod \bmax)$ has been discovered prior to the time slot $t$. However, this would imply that all other configurations $\left\{\,(c,b,t\bmod b)\mid b\in B\setminus\{\bmax\}\,\right\}$ have been discovered prior to the time slot $t$ as well. This is because they also send their beacons during the time slot $t'$, since $\bmax$ is an integer multiple of any $b \in B$. Consequently, if scanning a channel $c\in C$ in a time slot $t\in\{0,\ldots,\bmax|C|-1\}$ does not result in the discovery of the configuration $(c,\bmax,t \bmod \bmax)$, it does not result in any discoveries.

At the same time, for each time slot $t\in\{0,\ldots,\bmax|C|-1\}$ there is at least one channel $c\in C$ that results in the discovery of a configuration $(c,\bmax,t\bmod \bmax)$. If this would not be the case, all $|C|$ configurations $\left\{\,(c,\bmax,t\bmod\bmax)\mid c\in C\,\right\}$ would have had to be discovered prior to time slot $t$, that is, during one of the time slots $\left\{i\bmax+t\bmod\bmax\right\}_{i\in\{0,\ldots,\floor{t/\bmax}-1\}}$. The number of such time slots, however, is strictly less than $|C|$, leading to a contradiction.

Thus, since a \ALG{} algorithm maximizes the expected number of discoveries in each time slot, it would not select a channel that results in 0 discoveries if it can select a channel that results in at least one discoverable configuration. Consequently, in each time slot $t$ it discovers one configuration from the set $\left\{\,(c,\bmax,t \bmod \bmax)\mid c\in C\,\right\}$. Since this set contains $\bmax|C|$ configurations, it implies that it takes $\bmax|C|$ time slots in order to discover all of them.

Note, however, that in $\mathbb{F}_2$, a configuration $(c,b,\delta)$, $c\in C$, $b\in B$, $\delta\in\{0,\ldots,b-1\}$, is discovered if the configuration $(c,\bmax,\delta)$ is discovered, since in each time slot during which the configuration $(c,\bmax,\delta)$ sends its beacons, also the configuration $(c,b,\delta)$ sends its beacons.
\end{proof}

We remark that the assumption that the considered \ac{BP} set is in $\mathbb{F}_2$ is crucial to prove the \ac{WDT}-optimality, as shown by a counterexample in Example~\ref{exa:greedy_suboptimal}.

For the family of \ac{BP} sets $\mathbb{F}_3\subset\mathbb{F}_2$ we show the much stronger result that schedules generated by \ALG{} are not only complete and \ac{WDT}-optimal, but that they are recursive, and thus also \ac{MDT}-optimal and \ac{NDoT}-optimal. This result is formulated in the following proposition.

\begin{proposition}
\label{prop:greedy_f3}
For \ac{BP} sets from $\mathbb{F}_3$, a schedule is \ALG{} if and only if it is recursive. Consequently, for \ac{BP} sets $B\in\mathbb{F}_3$ \ALG{} schedules are complete, \ac{WDT}-optimal, \ac{MDT}-optimal, and \ac{NDoT}-optimal.
\end{proposition}
\begin{proof}
The claim that a recursive schedule is \ALG{} for \ac{BP} sets from $\mathbb{F}_3$ follows from Proposition~\ref{prop:recursive_optimal}, demonstrating that recursive schedules maximize the number of configurations detected until each time slot. 

It remains to show that a \ALG{} schedule is recursive. More precisely, we need to show that maximizing the expected number of discoveries in each time slot leads to discovering all configurations $(c,b,\delta)$, $\delta\in\{0,\ldots,b-1\}$, during the time slots $t\in\{0,\ldots,b|C|-1\}$.

Assume that this is not the case and that a configuration $(c,b,\delta)$, $c\in C$, $b\in B$, $\delta\in\{0,\ldots,b-1\}$, is discovered in a time slot $T_{(c,b,\delta)}\geq b|C|$. This implies that channel $c$ was not scanned during any of the time slots $\left\{\,ib+\delta\mid i\in\{0,\ldots,b-1\}\,\right\}$. Since there are $|C|$ such slots, this implies that there exists a different channel $c'\in C$ which was scanned during at least two of these time slots: $t',t''\in\left\{\,ib+\delta\mid i\in\{0,\ldots,b-1\}\,\right\}$. W.l.o.g., assume $t'<t''$. 

Consider the two subsets $B',B''\subseteq B$, with $B'=\left\{\,b'\in B\mid b'\leq b\,\right\}$, and $B''=\left\{\,b''\in B\mid b''> b\,\right\}$. Observe that by scanning channel $c$ during the time slot $t''$ at least the configurations $\left\{\,(c,b'',t''\bmod b'')\mid b''\in B''\,\right\}$ and the configuration $(c,b,t''\bmod b)$ are detected. However, by scanning channel $c'$ during the time slot $t''$, at most the configurations $\left\{\,(c',b'',t''\bmod b'')\mid b''\in B''\,\right\}$ are detected. The reason that none of the configurations with a \ac{BP} $b\in B'$ are detected during the time slot $t''$ is that they are all detected during the time slot $t'$, since $b$ is an integer multiple of each $b'\in B'$.

Consequently, scanning channel $c$ during time slot $t''$ would have resulted in larger expected number of discoveries than scanning channel $c'$. This leads to a contradiction since a \ALG{} algorithm maximizes the expected number of discoveries in each time slot.
\end{proof}

We remark that the assumption that the considered \ac{BP} set is in $\mathbb{F}_3$ is necessary to prove the \ac{MDT}-optimality of schedules generated by \ALG{}, as shown by a counterexample in Example~\ref{exa:greedy_suboptimal}. 

\section{Performance Analysis for CHAN TRAIN}
\label{sec:optim_chantrain}

\begin{restatable}{proposition}{chantraingreedy}
\label{prop:chantraingreedy}
For a \ac{BP} set $B\in\mathbb{F}_3$, CHAN TRAIN is \ALG{}.
\end{restatable}
\begin{proof}
Note that CHAN TRAIN does not perform a channel selection in each time slot but that it may skip several time slots if during those time slots the expected number of discoveries on the selected channel does not decrease. This may lead to a schedule that does not maximize the expected number of discoveries in each time slot, since it is possible that scanning a different channel during one of the skipped time slots would result in a higher value. This, however, would imply that the expected number of discoveries on that other channel has increased since the time slot when the current channel was selected. In the following, however, we show that for \ac{BP} sets from $\mathbb{F}_3$, the expected number of discoveries is monotonically decreasing over time.

W.l.o.g. we assume $B=\{b_0,\ldots,b_{n-1}\}$, $n\in\mathbb{N}_+$, $b_i>b_j$ for $i>j$.

Observe that since over $\mathbb{F}_3$ schedules obtained by \ALG{} are recursive, they discover all configurations with a \ac{BP} $b\in B$ during the first $b|C|$ time slots. Note that since there are $b|C|$ such configurations and it is not possible to discover two of them during one time slot (they either have different channels or different offsets), exactly one such configuration must be detected during each of the first $b|C|$ time slots. Consequently, the number of configurations discovered during each time slot $t\in\{b_i|C|,\ldots,b_{i+1}|C|-1\}$ is $n-(i+1)$. This means that the number of discoverable configurations is monotonically decreasing over time. Furthermore, since we assume that a neighbor selects a certain channel $c$ and beacon offset $\delta$ randomly according to a uniform distribution, it follows that the expected number of discoveries is also monotonically decreasing over time.
\end{proof}

Note that for a \ac{BP} set $B\not\in\mathbb{F}_3$ Proposition~\ref{prop:chantraingreedy} no longer holds, as illustrated by Example~\ref{exa:chantrain_nongreedy}. However, CHAN TRAIN is still \ac{WDT}-optimal on $\mathbb{F}_2$.

\begin{restatable}{proposition}{chantrainftwomakespanopt}
\label{prop:chantrainf2makespanopt}
For a \ac{BP} set $B\in\mathbb{F}_2$, CHAN TRAIN is \ac{WDT}-optimal.
\end{restatable}
\begin{proof}
The proof proceeds along the lines of the proof for Proposition~\ref{prop:greedy_f2} using the observation that, analogously to \ALG{} algorithms, CHAN TRAIN would never produce an idle slot as long as at least one discovery can be made.
\end{proof}

\section{Performance Analysis for \OPTBTwo{}}
\label{sec:optim_optbtwo}

\begin{restatable}[]{proposition}{propoptTwo}
\label{prop:opt2}
For arbitrary \ac{BP} sets with $|B|=2$, \OPTBTwo{} generates recursive listening schedules. Thus, they are complete, \ac{WDT}-optimal, \ac{MDT}-optimal, and \ac{NDoT}-optimal.
\end{restatable}
\begin{proof}
The claim follows directly from the definitions of \OPTBTwo{} and recursive schedules.
\end{proof}

\section{The \ac{WDT} of \ac{MDT}-Optimal Schedules}
\label{app:mdt_opt_cdt_opt}

In this section, we will show that for an arbitrary set of \acp{BP}, a schedule minimizing the \ac{MDT} uses at most $LCM(B)\left|C\right|$ time slots, and that, consequently, for \ac{BP} sets in $\mathbb{F}_2$, any \ac{MDT}-optimal schedule is also \ac{WDT}-optimal.

The following proposition establishes an upper bound on the number of time slots required to minimize an arbitrary strictly increasing function of discovery times. The idea for the proof is to show that any schedule that results in a configuration being detected after the time slot $LCM(B)\lvert C\rvert$ can be modified such that the configuration is detected before time slot $LCM(B)\lvert C\rvert$, without increasing the discovery times of other configurations.

\begin{restatable}{proposition}{lcmBC}
\label{prop:lcmBC}
For an arbitrary set of \acp{BP} $B\in\mathbb{F}_1$, a set of channels $C$, and a function $f:\mathbb{N}^{\lvert K_{BC}\rvert}\rightarrow\mathbb{R}$, which is strictly increasing in each argument, complete schedules $\mathcal{L}^*$ that minimize $f\left(\left(T_\kappa\left(\mathcal{L}\right)\right)_{\kappa\in K_{BC}}\right)$ have a \ac{WDT} $T_{\mathcal{L}^*}\leq LCM(B)\lvert C\rvert$.
\end{restatable}
\begin{proof}
Assume schedule $\mathcal{L}$ minimizes $f$ and $T_{\mathcal{L}}>LCM(B)\lvert C\rvert$. Consequently, there is at least one configuration $\kappa=\left(c,b,\delta\right)$ with discovery time $T_{\kappa}\left(\mathcal{L}\right)=T_{\mathcal{L}}>LCM(B)\lvert C\rvert$. Consider time slots $\tilde{\mathcal{T}}_\kappa=\left\{\delta+i LCM(B)\right\}_{i\in\left\{0,\ldots,\lvert C\rvert-1\right\}}$. Observe that $\left(\left\{c\right\}\times\tilde{\mathcal{T}}_\kappa\right)\cap\mathcal{L}=\emptyset$ since otherwise $\kappa$ would have been detected during one of the time slots in $\tilde{\mathcal{T}}_\kappa$. Consequently, there either exists an idle time slot $\tilde{t}\in\tilde{\mathcal{T}}_\kappa$, or, since $\lvert\tilde{\mathcal{T}}_\kappa\rvert=\lvert C\rvert$, there exist time slots $t',t''\in\tilde{\mathcal{T}}_\kappa$, $t'\neq t''$, and a channel $c'\neq c$ such that $\left(c',t'\right),\left(c',t''\right)\in\tilde{\mathcal{T}}_\kappa$.

In the first case, we construct a new schedule $\mathcal{L}'=\mathcal{L}\setminus\left\{\left(c,\,T_{\mathcal{L}}\right)\right\}\cup\left\{\left(c,\,\tilde{t}\right)\right\}$, such that $\kappa$ is detected during $\tilde{t}$ and none of the discovery times of other neighbor configurations are increased.

In the second case, we construct a new schedule $\mathcal{L}'=\mathcal{L}\setminus\left\{\left(c,\,T_{\mathcal{L}}\right),\left(c',\,t''\right)\right\}\cup\left\{\left(c,\,t''\right)\right\}$. With the new schedule, configuration $\kappa$ is detected during time slot $t''$. In order to show that the discovery times of other neighbors do not increase, consider the vector function $\delta(t)=\left(t\bmod b,\,b\in B\right)$, providing for each time slot $t$ a vector of offsets that can be detected in $t$. Since periodicity of $\delta(t)$ is $LCM(B)$, we conclude that $\delta(t')=\delta(t'')$, and thus no discoveries are performed during time slot $t''$ with the schedule $\mathcal{L}$. Consequently, none of the discovery times are increased with the new schedule.

Repeating the above procedure for each $\kappa$ with discovery time $T_{\kappa}\left(\mathcal{L}\right)>LCM(B)\lvert C\rvert$ results in a schedule $\mathcal{L}^*$ with \ac{WDT} $T_{\mathcal{L}^*}\leq LCM(B)\lvert C\rvert$ with $f\left(\left(T_\kappa\left(\mathcal{L}^*\right)\right)_{\kappa\in K_{BC}}\right)<f\left(\left(T_\kappa\left(\mathcal{L}\right)\right)_{\kappa\in K_{BC}}\right)$, proving the claim.
\end{proof}

The following Corollary presents a notable consequence from Proposition~\ref{prop:lcmBC} for \ac{BP} sets from $\mathbb{F}_2$.

\begin{restatable}{corollary}{lcmBCcorOne}
\label{cor:lcmBC}
For a \ac{BP} set $B\in\mathbb{F}_2$, a set of channels $C$, and a function $f:\mathbb{N}^{\lvert K_{BC}\rvert}\rightarrow\mathbb{R}$, which is strictly increasing in each argument, complete schedules $\mathcal{L}^*$ that minimize \newline $f\left(\left(T_\kappa\left(\mathcal{L}\right)\right)_{\kappa\in K_{BC}}\right)$  are \ac{WDT}-optimal.
\end{restatable}
\begin{proof}
The claim follows directly from the definition of \ac{WDT}-optimality, Proposition~\ref{prop:lcmBC}, and the defining property of $B\in\mathbb{F}_2$ that $LCM(B)=\max(B)$.
\end{proof}

Please observe that the fact that \ac{MDT} satisfies the conditions on the function $f$ in Proposition~\ref{prop:lcmBC} implies that the upper bounds established in Proposition~\ref{prop:lcmBC} and Corollary~\ref{cor:lcmBC} also apply to schedules minimizing \ac{MDT}.

\begin{restatable}{corollary}{lcmBCcorOnemdt}
\label{cor:mdtcor1}
For an arbitrary set of \acp{BP} $B\in\mathbb{F}_1$, and a set of channels $C$, \ac{MDT}-optimal listening schedules $\mathcal{L}^*$ have a \ac{WDT} $T_{\mathcal{L}^*}\leq LCM(B)\lvert C\rvert$.
\end{restatable}
\begin{proof}
Observe that \ac{MDT} is strictly increasing in each configuration detection time $T_\kappa$. Applying Proposition~\ref{prop:lcmBC} we obtain the claim.
\end{proof}

\begin{restatable}{corollary}{lcmBCcorTwo}
\label{cor:mdtcor2}
For a set of \acp{BP} $B\in\mathbb{F}_2$, and a set of channels $C$, \ac{MDT}-optimal listening schedules are also \ac{WDT}-optimal.
\end{restatable}
\begin{proof}
Observe that \ac{MDT} is strictly increasing in each configuration detection time $T_\kappa$. Applying Corollary~\ref{cor:lcmBC} we obtain the claim.
\end{proof}

\section{\ac{NDoT}-Optimality}
\label{app:mdt_opt_ndot_opt}

\begin{restatable}{proposition}{ndot-cdt}
\label{prop:ndot-cdt}
A \ac{NDoT}-optimal schedule $\mathcal{L}$ is also \ac{WDT}-optimal.
\end{restatable}
\begin{proof}
Since a \ac{NDoT}-optimal schedule pointwise maximizes the \ac{CDF} of discovery times, the value of the \ac{CDF} for the normalized discovery time 1 is also maximized. Since there always exist schedules that discover all neighbors at the normalized discovery time 1, and therefore whose \acp{CDF} at normalized time 1 have the value 1, we conclude that the \ac{CDF} of a \ac{NDoT}-optimal schedule also must have the value 1 at the normalized time 1.
\end{proof}

\begin{restatable}{proposition}{mdt-opt_ndot-opt}
\label{prop:mdt-opt_ndot-opt}
If for a set of \acp{BP} $B\in\mathbb{F}_1$ and a set of channels $C$ a \ac{NDoT}-optimal schedule $\mathcal{L}^*$ exists, it follows that any \ac{MDT}-optimal schedule $\mathcal{L}$ is also \ac{NDoT}-optimal.
\end{restatable}
\begin{proof}
The proof is constructed by contradiction. Assume that schedule $\mathcal{L}^*$ is \ac{NDoT}-optimal. It implies that $\mathcal{L}^*$ is also \ac{MDT}-optimal. Furthermore, assume that another schedule $\mathcal{L}$ is \ac{MDT}-optimal but not \ac{NDoT}-optimal. The \acp{CDF} of normalized discovery times of $\mathcal{L}$ and $\mathcal{L}^*$ shall be denoted by $f^*$ and $f$, respectively. Since both schedules are \ac{MDT}-optimal we obtain:
\begin{equation*}
\int_{0}^{T^*} \! x f^{*'}(x) \, dx = \int_0^T \! x f^{'}(x) \, dx\,,
\end{equation*}
where $T^*,T\geq 1$ denote the respective \acp{WDT}.

Applying partial integration and using (i) the fact that $T^*=1$ and $f^*(1)=1$ due to the \ac{WDT}-optimality of $\mathcal{L}^*$, and (ii) the fact that $f(0)=f^*(0)=0$, we obtain

\begin{align*}
& x f^{*}(x) \bigg\rvert_{0}^{1} - \int_{0}^{1} \!f^{*}(x)\,dx - x f(x) \bigg\rvert_{0}^{b} + \int_{0}^{b} \!f(x) \, dx  = 0 \\
& 1 - \int_{0}^{1} \!f^{*}(x)\,dx - b + \int_{0}^{b} \!f(x) \, dx  = 0 \\
\Rightarrow\, & 1 - b - \int_{0}^{1} \!f^{*}(x) \, dx + \int_{0}^{1} \!f(x) \, dx + \int_{1}^{b} \!f(x) \, dx = 0 \\
\Rightarrow\, & (b - 1) + \left[\int_{0}^{1} \!f^{*}(x) \, dx -  \int_{0}^{1} \!f(x) \, dx \right] - \int_{1}^{b} \!f(x) \, dx = 0 
\end{align*}

On the other hand, since $\mathcal{L}$ is not \ac{NDoT}-optimal: $\int_{0}^{1} \!f^{*}(x) \, dx -  \int_{0}^{1} \!f(x) \, dx > 0$. Furthermore, observe that $\int_{1}^{b} \!f(x) \, dx \leq b-1$. We obtain a contradiction proving the original claim.
\end{proof}

\section{Computational Complexity and Memory Requirements}
\label{sec:complexity}

Note that \ALG{} algorithms have a compelling property of low complexity. A possible implementation of pseudo code in Algorithm~1 presented in \cite{Karowski18} proceeds as follows. For each configuration it stores a binary variable indicating if it has been covered or not, resulting in $\left|C\right|\left|B\right|\sum_{b\in B}b$ bits of required memory space. At each time slot $t$, it iterates over all channels $c\in C$, computing for each channel, which of configurations $\left\{\,(c,b,t\bmod b)\mid b\in B\,\right\}$ are not yet considered. Finally, it selects a channel, for which the sum of probabilities of these configurations to be selected by a neighbor is the highest. Consequently, computational complexity at each time slot is $\mathcal{O}\left(\left|C\right|\left|B\right|\right)$. It iterates over time slots until all configurations are covered. 

The overall computational complexity depends on the number of time slots required to consider all configurations. Since \alg{} algorithms are \ac{WDT}-optimal for $B\in\mathbb{F}_2$, the number of required time slots is $\max(B)|C|$, resulting in a total complexity of $\mathcal{O}\big(|C|^2|B|\max(B)\big)$. For \ac{BP} sets from $\mathbb{F}_1$, \ALG{} algorithms require at most $LCM(B)\lvert C\rvert$. A proof of this claim is similar to the proof of Proposition~\ref{prop:greedy_f2} and is omitted for brevity.) Thus, an upper bound for the computational complexity over $\mathbb{F}_1$ is $\mathcal{O}\big(|C|^2|B|LCM(B)\big)$.

Please note that if the assumption made in our system model in Section III in~\cite{Karowski18} that the \ac{GCD} of the considered \ac{BP} sets is 1 does not hold, the complexity can be further reduced by replacing $B$ by $B'=\left\{\frac{b}{GCD(B)}\right\}_{b\in B}$. This preprocessing step allows to reduce computational complexity over $\mathbb{F}_2$ to $\mathcal{O}\left(\left|C\right|^2\left|B\right|\frac{\max(B)}{GCD(B)}\right)$, over $\mathbb{F}_1$ the upper bound becomes $\mathcal{O}\left(\left|C\right|^2\left|B\right|\frac{LCM(B)}{GCD(B)}\right)$.

\begin{figure*}
\centering
\subfloat[\acf{SMDT} $\mathbb{F}_2$] {
\includegraphics[width=\evalFigWidth\textwidth]{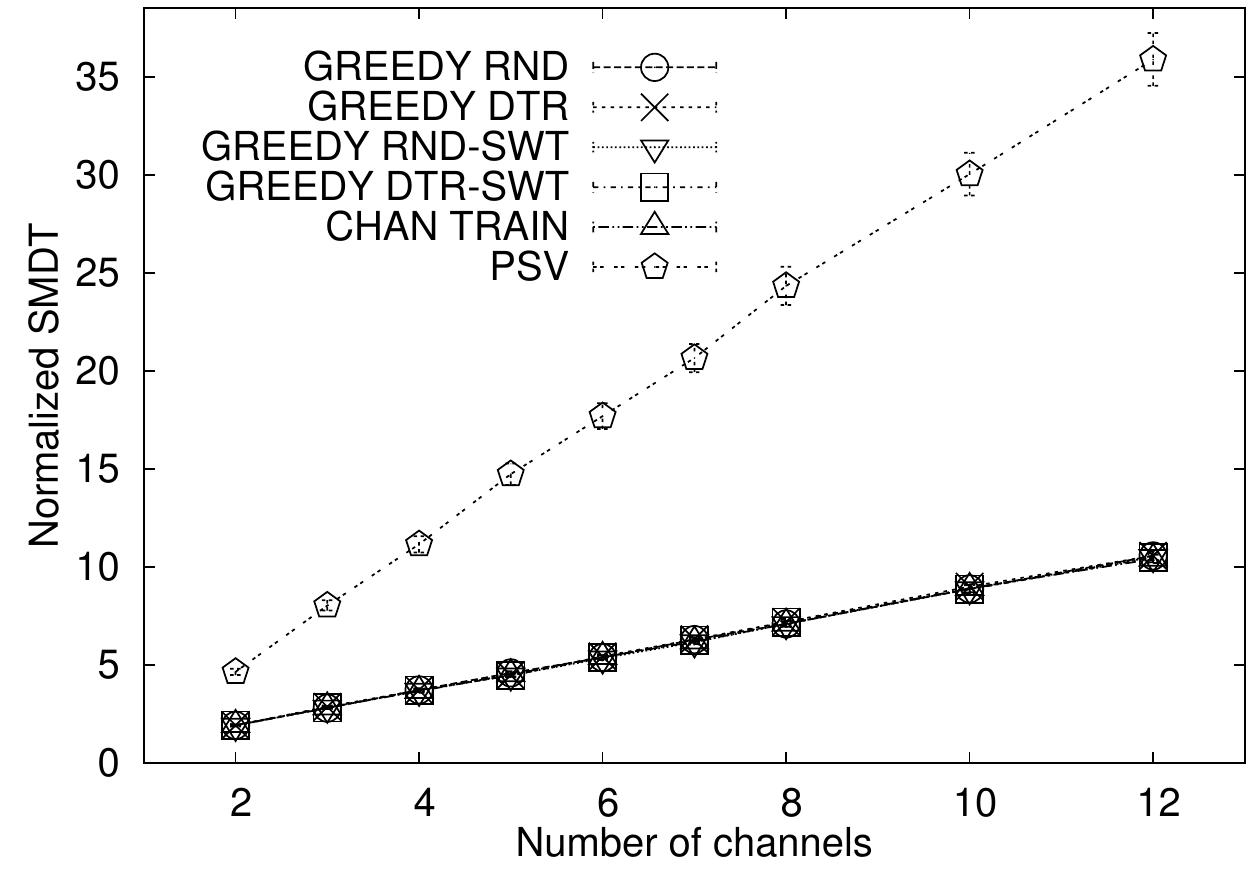}
\label{fig:F2_numChan_discTime_sim}
}
\subfloat[\acf{SMDT} $\mathbb{F}_2$]{
\includegraphics[width=\evalFigWidth\textwidth]{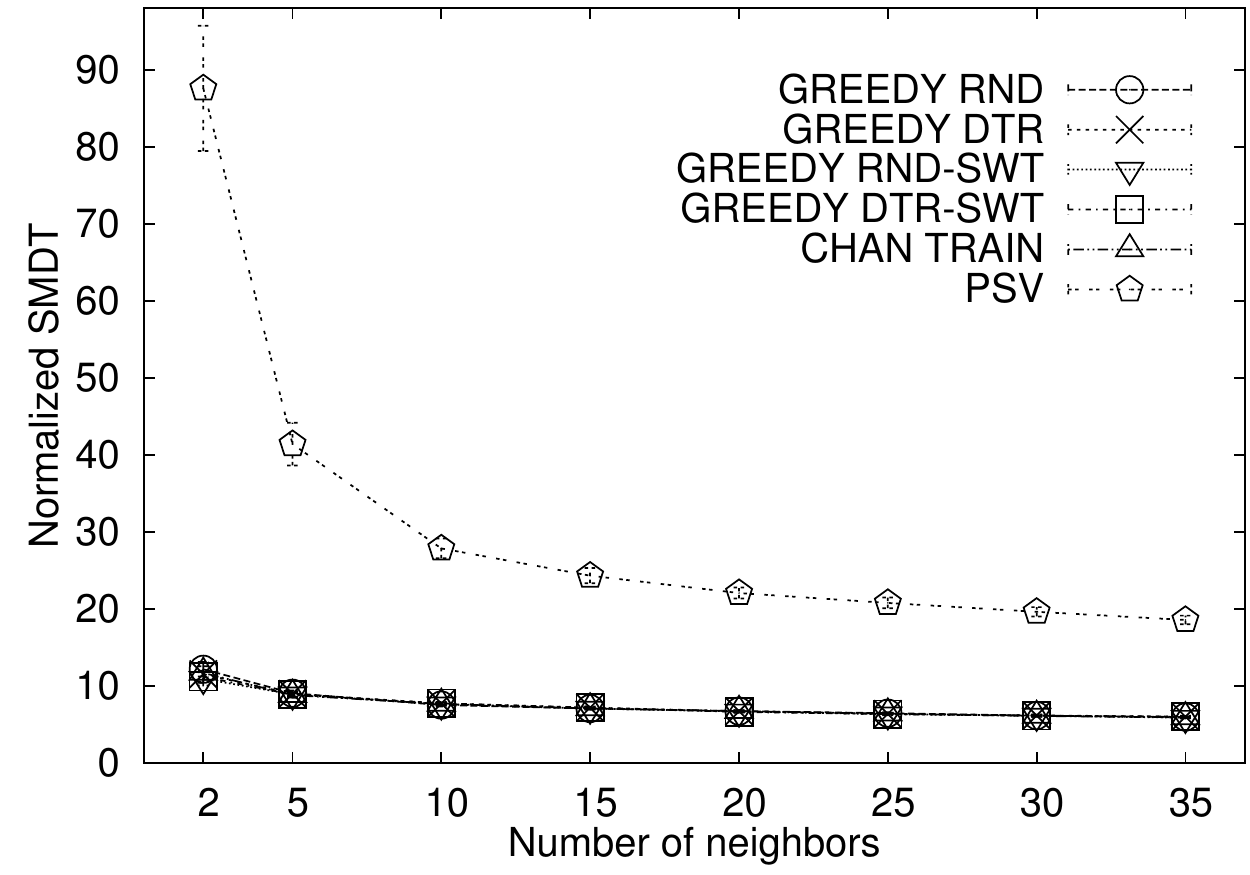}
\label{fig:F2_numConf_discTime_sim}
} \\
\subfloat[\acf{SNDoT} $\mathbb{F}_2$]{	
\includegraphics[width=\evalFigWidth\textwidth]{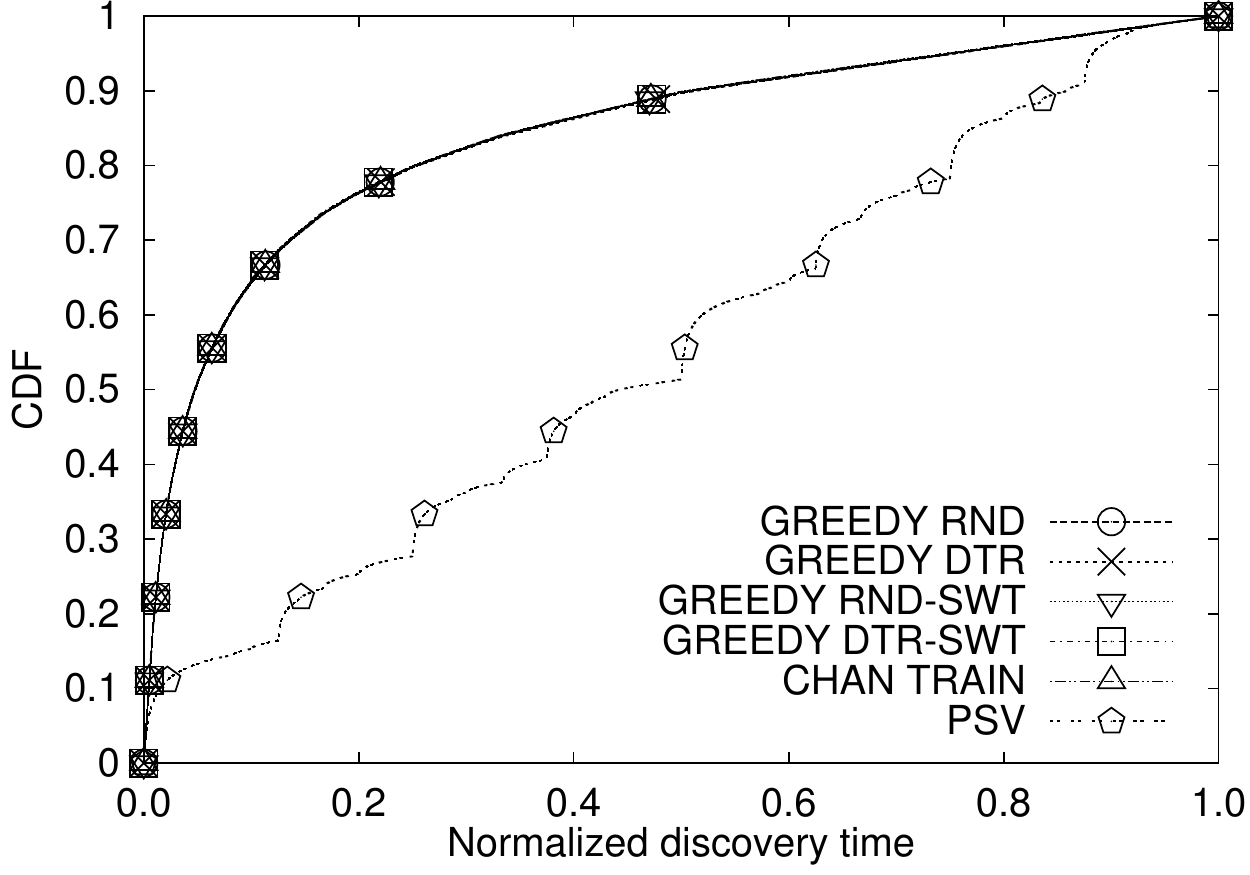}
\label{fig:F2_cdf_sim}
}
\caption{Evaluation results obtained by simulation for the family of \ac{BP} sets $\mathbb{F}_2$ (see Section~\ref{subsec:sim_results_appendix_f2}).}
\label{fig:F2_sim_results_1}
\end{figure*}

\begin{figure}
\centering
\subfloat[\acf{SWDT} $\mathbb{F}_2$]{
\includegraphics[width=\evalFigWidth\textwidth]{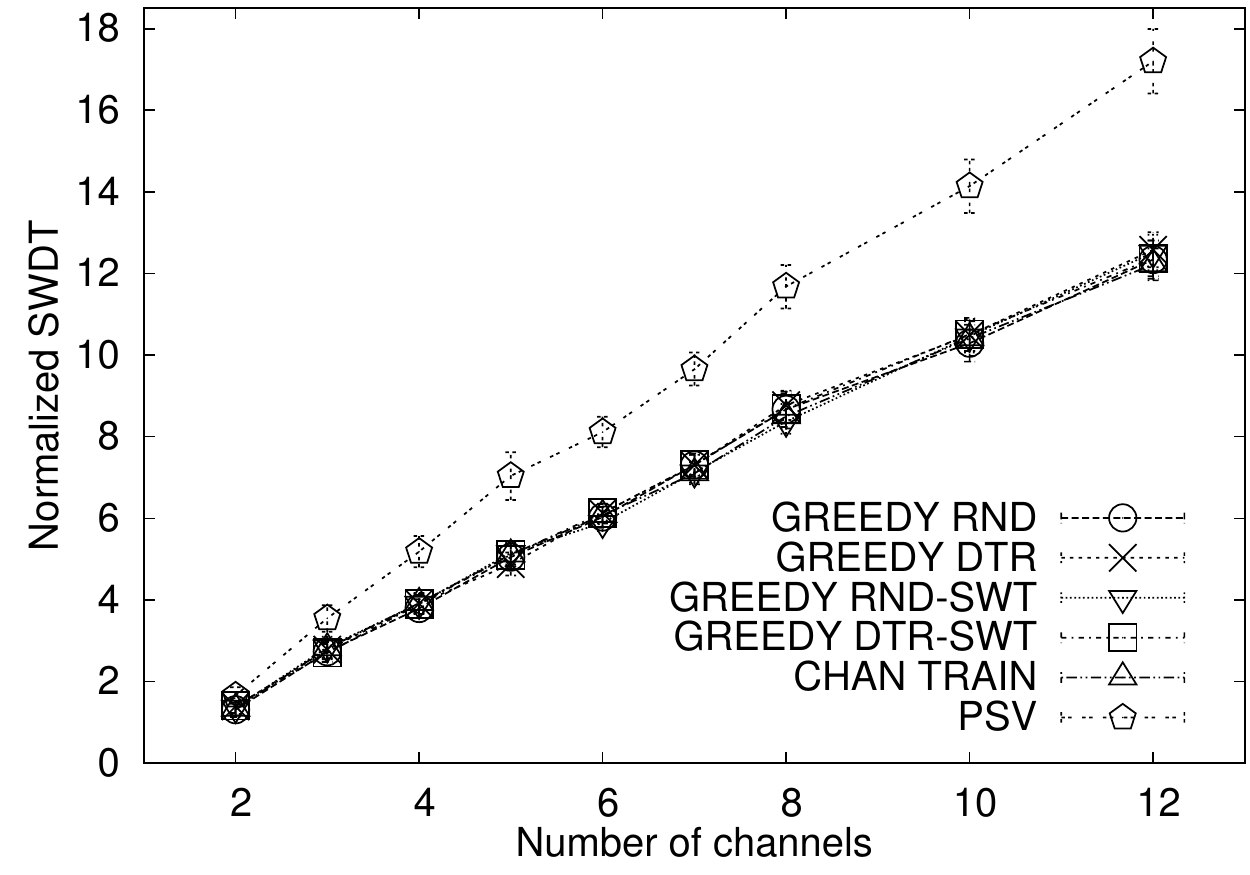}
\label{fig:F2_numChan_makespan_sim}
}
\subfloat[\acf{SWDT} $\mathbb{F}_2$]{
\includegraphics[width=\evalFigWidth\textwidth]{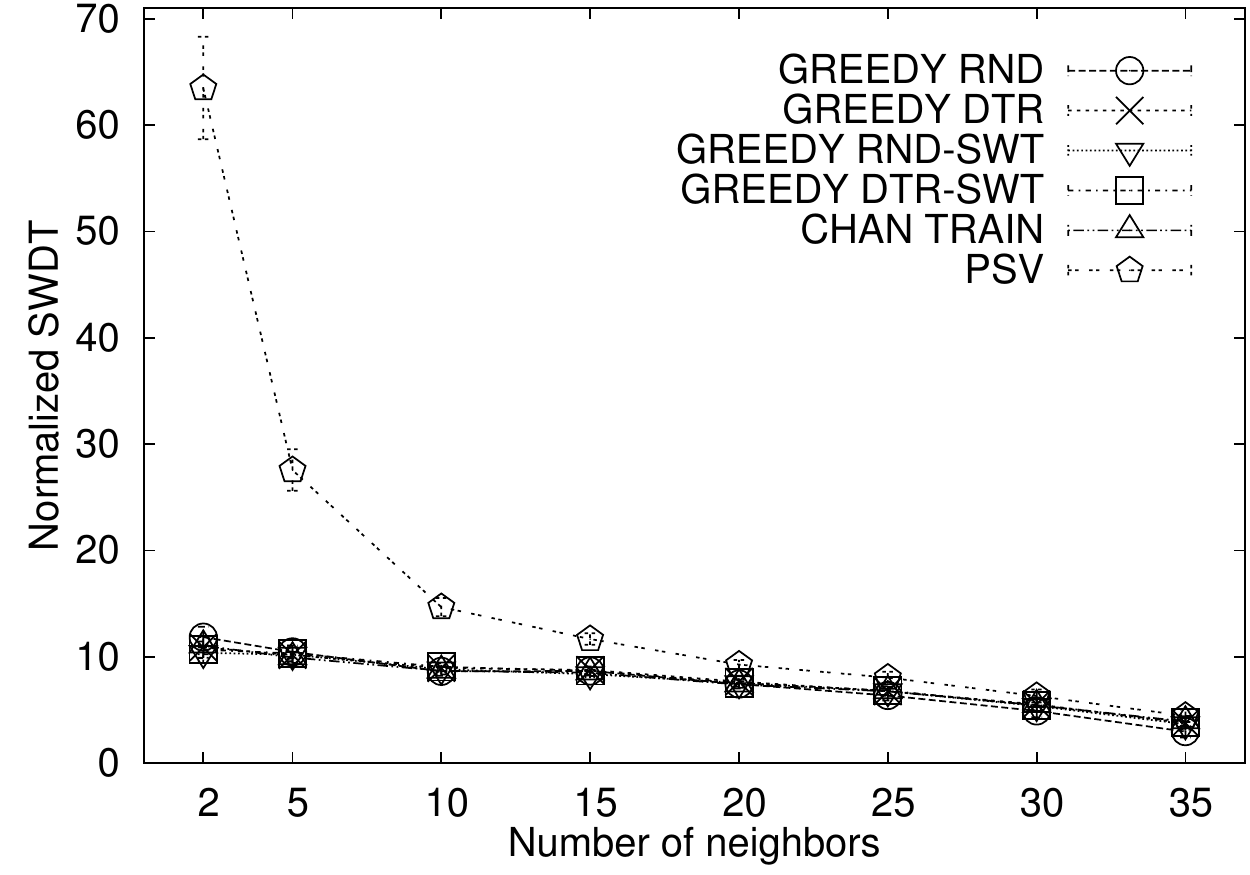}
\label{fig:F2_numConf_makespan_sim}
}
\caption{Evaluation results obtained by simulation for the family of \ac{BP} sets $\mathbb{F}_2$ (see Section~\ref{subsec:sim_results_appendix_f2}).}
\label{fig:F2_sim_results_2}
\end{figure}

\section{Additional simulation results}
\label{sec:appendix_results}

In the following we present supplementary simulation results evaluating three performance metrics that are only defined for specific instances of the studied problem, and are therefore named \acf{SMDT}, \acf{SNDoT} and \acf{SWDT}. These metrics are more precisely specified in the following. The setting is described in Section~IX-A in~\cite{Karowski18}.

\subsection{Performance Metrics}

In the following we evaluate three performance metrics: \ac{SNDoT}, \ac{SMDT}, and \ac{SWDT}. The \ac{SMDT} is defined as $\frac{1}{|N|} \sum_{\nu\in N}T_\nu(\mathcal{L})$. It relates to the \ac{MDT} in the way a sample mean relates to the 'true' mean. Since a sampled mean is an unbiased estimator, the expected value of the \ac{SMDT} is the \ac{MDT}. 
The \ac{SWDT} is defined as $\max_{\nu\in N} T_\nu(\mathcal{L})$. It is the analogon of the \ac{WDT} in a specific network environment, where only a subset of possible configurations is present. However, while the \ac{SMDT} can be interpreted as an estimation of \ac{MDT}, \ac{SWDT} has no such relationship with the \ac{WDT}. In fact, the \ac{WDT} is a more important property of a listening schedule, since it determines the required execution time in scenarios, where a complete discovery is desired. The \ac{SWDT}, on the other hand, cannot be used as an indicator for the discoverer to stop the discovery process, since the discoverer cannot know if all neighbors have been discovered or not, before it has probed for all potential configurations, which requires an execution time equal to the \ac{WDT}. Still, when the number of neighbors increases, \ac{SWDT} converges to the \ac{WDT} due to the law of large numbers. Finally, \ac{SNDoT} is the number of discovered neighbors divided by the total number of neighbors as a function of the execution time of the schedule. Also for the \ac{SNDoT} the law of large numbers implies convergence to the \ac{CDF} of discovery times for large numbers of neighbors.

The discovery times used in \ac{SNDoT} are normalized as described in Section~IV in~\cite{Karowski18}. \ac{SMDT} and \ac{SWDT} are normalized w.r.t.\@ their corresponding optimum values obtained by solving an \ac{ILP} similar to \MDTOPT{} but with a different objective function. For each scenario with a given set of neighbors and their configurations two \acp{ILP} are solved that generate listening schedule minimizing the \ac{SMDT} and \ac{SWDT} among those neighbors. We use the following optimization variables.

\begin{figure*}
\centering
\subfloat[\acf{SMDT} $\mathbb{F}_1$]{
\includegraphics[width=\evalFigWidth\textwidth]{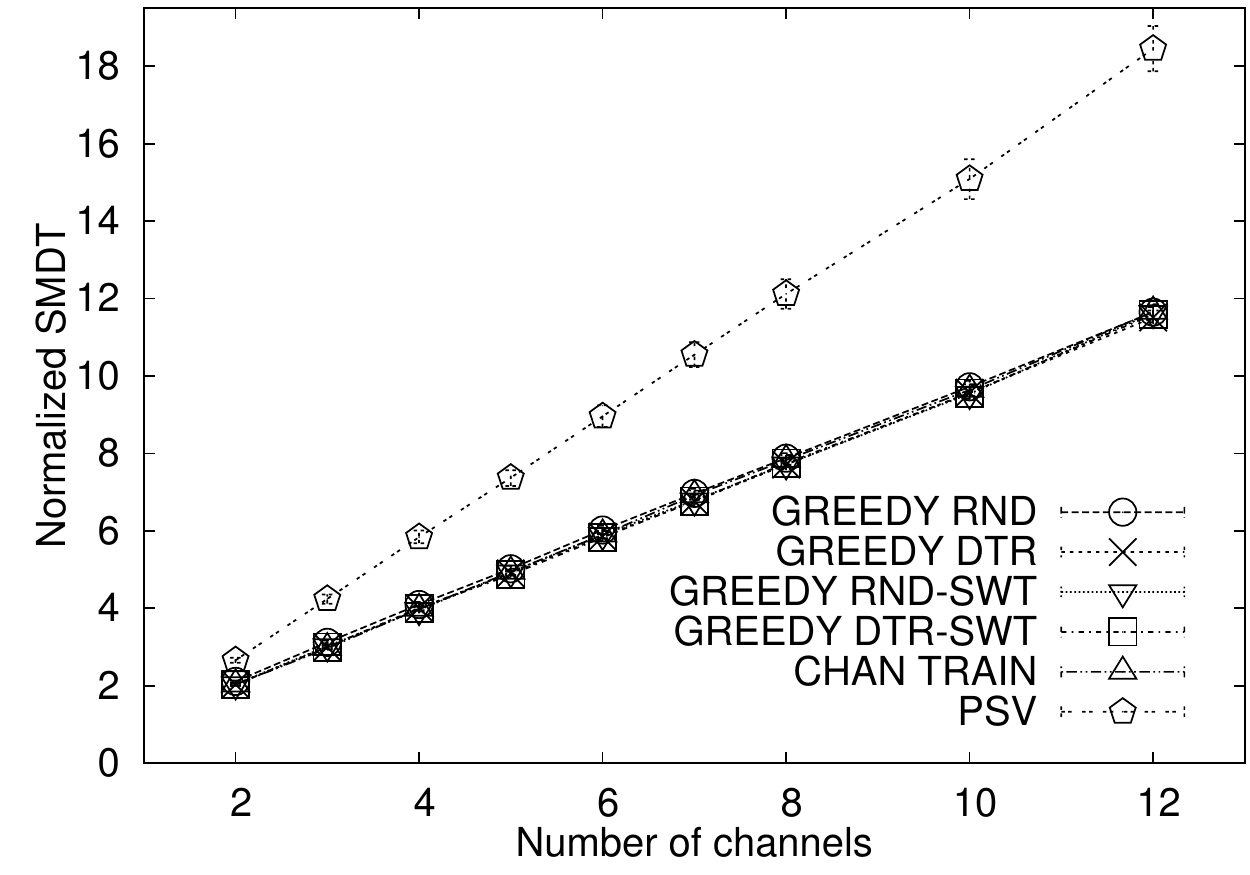}
\label{fig:F1_numChan_discTime_sim}
}
\subfloat[\acf{SMDT} $\mathbb{F}_1$]{
\includegraphics[width=\evalFigWidth\textwidth]{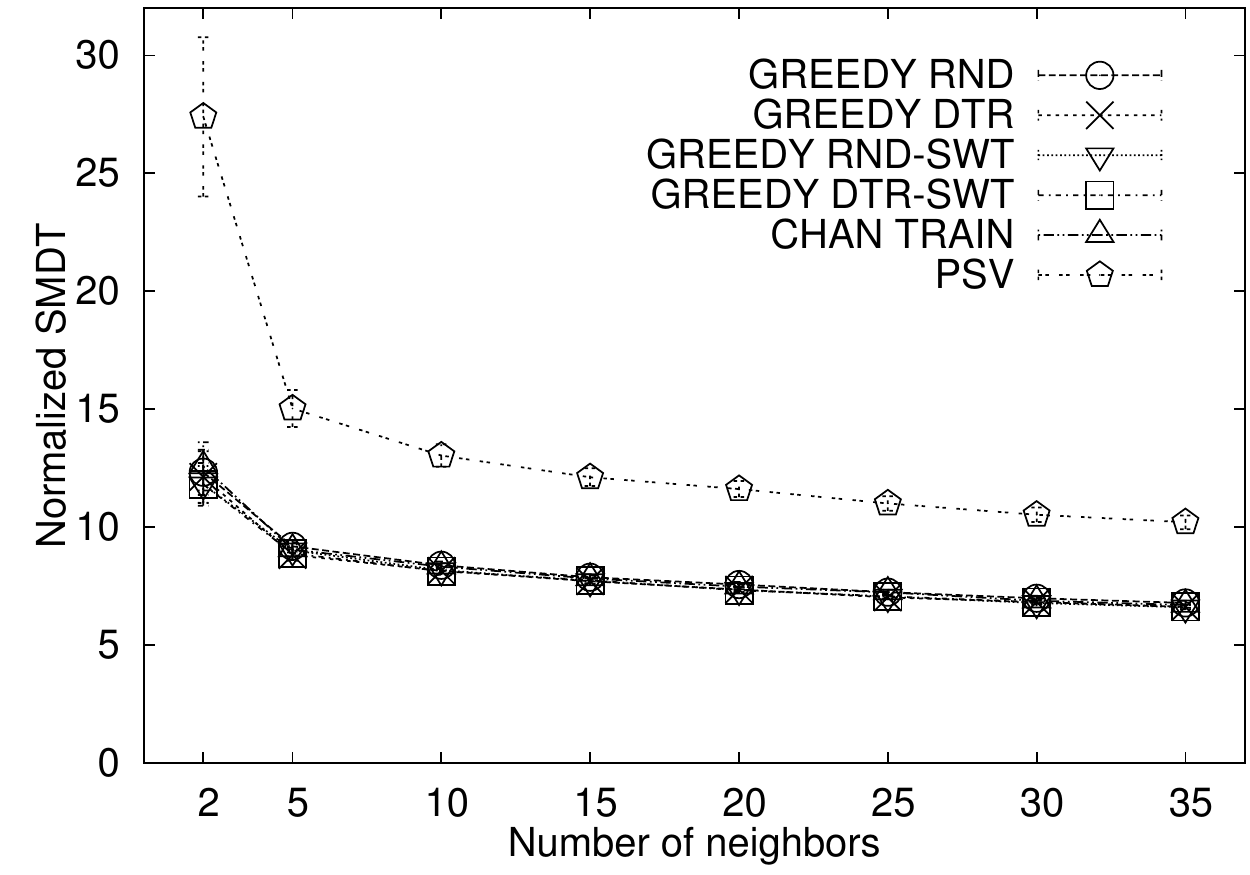}
\label{fig:F1_numConf_discTime_sim}
} \\
\subfloat[\acf{SNDoT} $\mathbb{F}_1$]{	
\includegraphics[width=\evalFigWidth\textwidth]{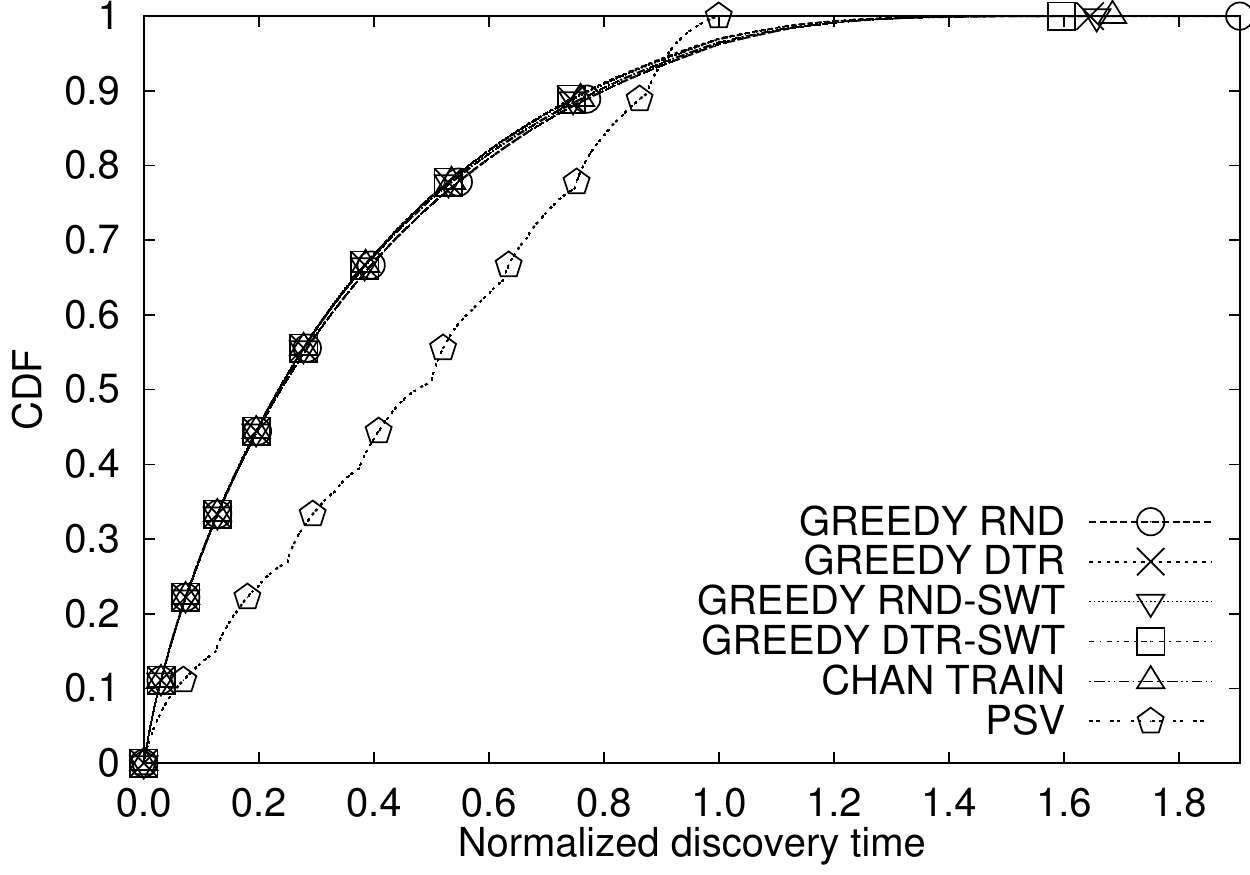}
\label{fig:F1_mdot_sim}
}
\caption{Evaluation results obtained by simulation for the family of \ac{BP} sets $\mathbb{F}_1$ (see Section~\ref{subsec:sim_results_appendix_f1}).}
\label{fig:F1_sim_results_1}
\end{figure*}

\begin{figure}
\centering
\subfloat[\acf{SWDT} $\mathbb{F}_1$]{
\includegraphics[width=\evalFigWidth\textwidth]{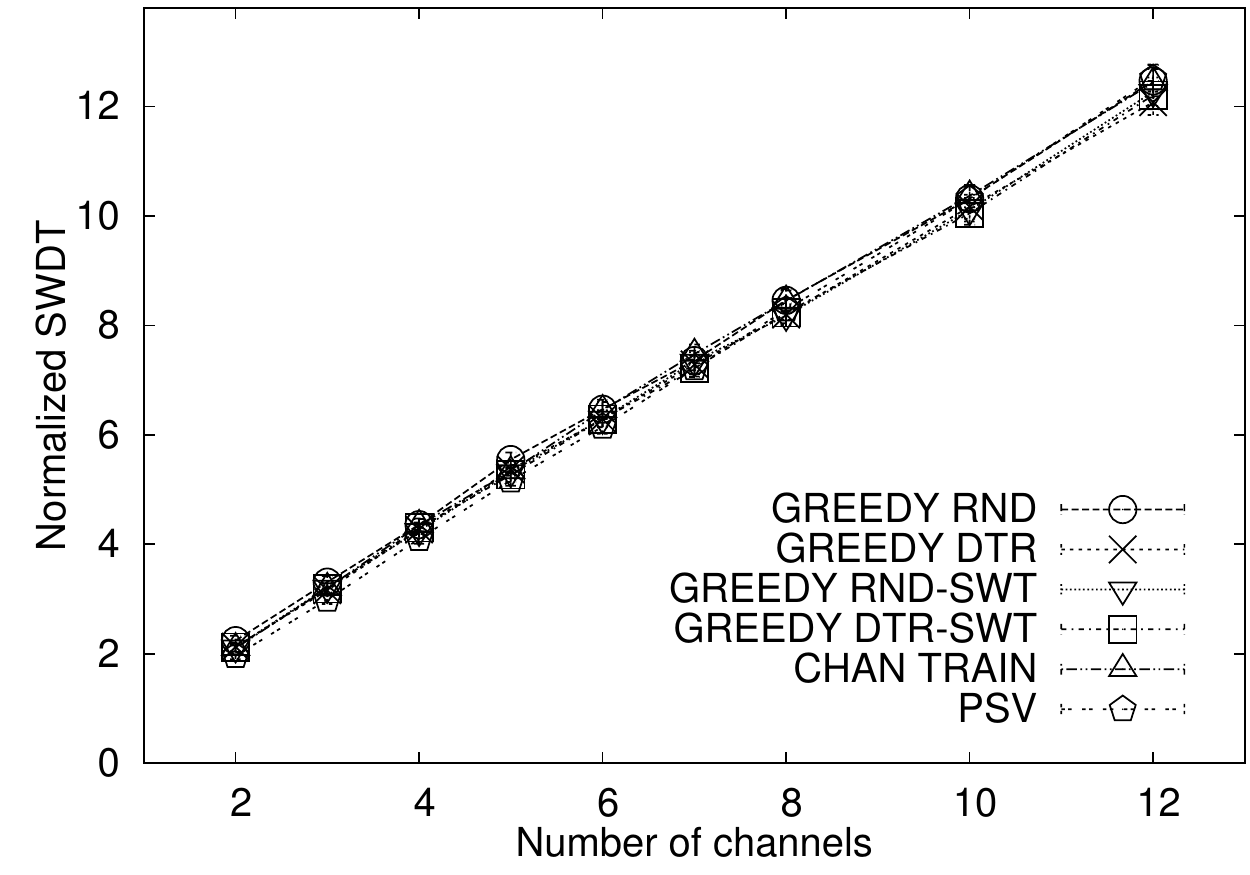}
\label{fig:F1_numChan_makespan_sim}
}
\subfloat[\acf{SWDT} $\mathbb{F}_1$]{
\includegraphics[width=\evalFigWidth\textwidth]{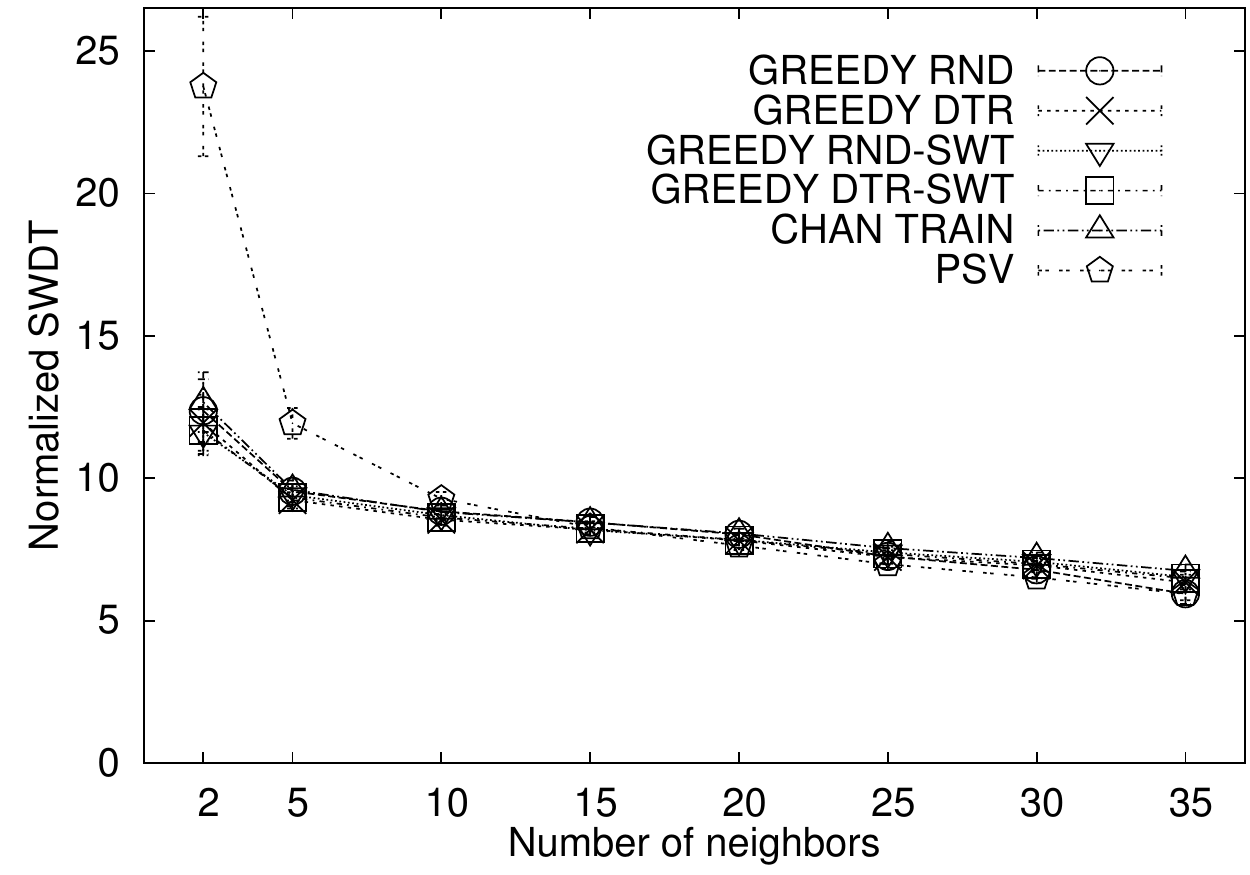}
\label{fig:F1_numConf_makespan_sim}
}
\caption{Evaluation results obtained by simulation for the family of \ac{BP} sets $\mathbb{F}_1$ (see Section~\ref{subsec:sim_results_appendix_f1}).}
\label{fig:F1_sim_results_2}
\end{figure}

\begin{eqnarray*}
\begin{aligned}
y_{i\nu} &= \begin{cases}
				\begin{aligned}
					 \textrm{1  ,} & \textrm{ if neighbor } \nu \in N  \textrm{ is detected during time slot }\\
												 & i b_\nu + \delta_\nu \\
           \textrm{0  ,} & \textrm{ otherwise}  
				\end{aligned}					
        \end{cases} \\
h_{ct} &= \begin{cases}
				\begin{aligned}
           \textrm{1  ,} & \textrm{ if channel } c \textrm{ is scanned during time slot } t \\
           \textrm{0  ,} & \textrm{ otherwise}  
				\end{aligned}	
        \end{cases}
\end{aligned}
\end{eqnarray*}

We use the following constraints.

		\begin{align*}
& \sum_{i=0}^{\floor{\frac{t_{max} - \delta_\nu}{b_\nu}}} y_{i\nu} = 1 &&\text{for all}\;\nu \in N  \label{eq:eval_C4}\tag{\~C1}\\
& y_{i\nu} \leq h_{c_{\nu},\delta_\nu + ib_\nu} &&\text{for all}\;\nu \in N, \\
& && i \in \left\{0, \ldots, \floor{\frac{t_{max} - \delta_\nu}{b_\nu}} \right\} \label{eq:eval_C5}\tag{\~C2}\\
& \sum_{c \in C} h_{ct} \leq 1 &&\text{for all}\; t \in \{0, \ldots, t_{max}\} \label{eq:eval_C6}\tag{\~C3}
\end{align*}

Constraint~\ref{eq:eval_C4} ensures that for each neighbor $\nu \in N$ a scan is scheduled in exactly one of the beacon time slots $\mathcal{T}_\nu$. The last two constraints~\ref{eq:eval_C5} and~\ref{eq:eval_C6} make sure that during each time slot not more than one channel is scanned. In comparison to the numerical evaluation we allow for much larger \acp{BP} from $B\in\mathbb{F}_1$. Therefore, to keep the \acp{ILP} computationally feasible, the maximum time slot used in the optimization problems is limited to $t_{max} = \min(1000 \max(B)|C| - 1, \text{LCM}(B)|C| - 1)$.

We use the following objective function to find the optimum \ac{SMDT} for a given network environment.
\begin{align*} 
\min\;\frac{1}{\vert N\vert} \sum_{\nu \in N}\sum_{i=0}^{\floor{\frac{t_{max} - \delta_\nu}{b_\nu}}} y_{i\nu} \left(ib_\nu + \delta_\nu\right)\,.
\end{align*}

We use the following objective function, together with an additional constraint, to find the optimum \ac{SWDT} for a given network environment.

\begin{align*}
\min\quad     &z &\notag   \\
\text{s.t.}\quad   &y_{i\nu} \left(ib_\nu + \delta_\nu\right) \leq z \\
&\text{for all}\;\nu \in N,\; i \in \left\{0, \ldots, \floor{ \frac{t_{max} - \delta_\nu}{b_\nu}} \right\}\,.
\end{align*}

Note that the results for the \ac{SMDT} and \ac{SWDT} presented in the following cannot be directly compared with the numerical results of the \ac{MDT} and \ac{WDT} presented in Section~VIII-B in~\cite{Karowski18} due to different methods of normalization.

\subsection{Results}
\label{subsec:sim_results_appendix}

In the following, we present simulations results for \ac{BP} families $\mathbb{F}_1$ and $\mathbb{F}_2$ accompanied by confidence intervals for a confidence level of 95\%.

\subsubsection{Results for $\mathbb{F}_2$}
\label{subsec:sim_results_appendix_f2}

Figures~\ref{fig:F2_numChan_discTime_sim} and~\ref{fig:F2_numConf_discTime_sim} depict the results for the normalized \ac{SMDT}, as a function of the number of channels and of the number of neighbors, respectively.
For two channels the PSV strategy results in a \ac{SMDT} which is about twice as high as the results of the \ALG{} approaches and CHAN TRAIN. However, the gap  is increasing with larger number of channels. For twelve channels the \ac{SMDT} of PSV is more than 350\% higher compared to the results of \ALG{} / CHAN TRAIN. For all analyzed number of neighbors PSV results again in a \ac{MDT} which is at least 350\% higher than the results achieved by \ALG{} and CHAN TRAIN.  In addition, all strategies diverge from the optimum with an increasing number of channels, while they approach the optimum with an increased number of neighbors.

Figure~\ref{fig:F2_cdf_sim} shows the results for \ac{SNDoT}. \ALG{} and CHAN TRAIN discover up to 50\% more neighbors in average at certain points in time without having an impact on the discovery time of the last percentage of discovered neighbors. 

Figures~\ref{fig:F2_numChan_makespan_sim} and~\ref{fig:F2_numConf_makespan_sim} display the normalized \acp{SWDT}. With two channels all strategies result in similar values. However, for larger number of channels, and even more so for low numbers of neighbors, the PSV strategy results in a significantly larger \ac{SWDT}. Again, we remark that the \ac{SWDT} is of a lower importance than the \ac{WDT} studied in Section~VIII in~\cite{Karowski18} (see also Section~IV).

\subsubsection{Results for $\mathbb{F}_1$}
\label{subsec:sim_results_appendix_f1}

Figures~\ref{fig:F1_sim_results_1} and ~\ref{fig:F1_sim_results_2} show the results for the family of \ac{BP} sets $\mathbb{F}_1$. The findings are qualitatively mostly similar to those for \ac{BP} sets from $\mathbb{F}_2$.

The normalized \ac{SMDT} as a function of the number of channels is depicted in Figures~\ref{fig:F1_numChan_discTime_sim} and~\ref{fig:F1_numConf_discTime_sim}. Similar to the results for $\mathbb{F}_2$ our approaches significantly outperform PSV. However, the gap between them has decreased.

Figure~\ref{fig:F1_mdot_sim} shows the results of the \ac{SNDoT}. The majority of neighbors is discovered faster by the \ALG{} strategies and CHAN TRAIN. However, about 10\% of neighbors are discovered later by \ALG{} and CHAN TRAIN as compared to PSV.

The normalized \ac{SWDT} is shown in Figures~\ref{fig:F1_numChan_makespan_sim} and~\ref{fig:F1_numConf_makespan_sim}. In contrast to the results for $\mathbb{F}_2$ the PSV strategy achieves similar results as our strategies for all analyzed number of channels. 

\subsection{Summary}

The simulation results reveal that even under realistic conditions the proposed algorithms significantly outperform the \ac{PSV} strategy regarding the \ac{SMDT} for the family of \ac{BP} sets $\mathbb{F}_2$ as well as $\mathbb{F}_1$ and \ac{SWDT} for $\mathbb{F}_2$ for different number of channels and number of neighbors. 
The results for \ac{SNDoT} are similar to the numerical results presented in Section~VIII-B in~\cite{Karowski18}. For the family of \ac{BP} sets $\mathbb{F}_2$ our strategies discover up to 50\% more neighbors compared to PSV at certain points in time without penalizing the discovery of the last neighbors. However, for the family of \ac{BP} sets $\mathbb{F}_1$ the gap decreases and the discovery of the last 10\% of neighbors requires additional time by our approaches in comparison with PSV.

\acrodef{AP}{Access Point} \acrodefplural{AP}[AP's]{Access Points}
\acrodef{BP}{Beacon Period} \acrodefplural{BP}[BP's]{Beacon Periods}
\acrodef{BO}{Beacon Order} \acrodefplural{BO}[BO's]{Beacon Orders}
\acrodef{WDT}{Worst-Case Discovery Time} \acrodefplural{WDT}[WDT's]{Worst-Case Discovery Times}
\acrodef{SWDT}{Sample Worst-Case Discovery Time} \acrodefplural{SCDT}[SWDT's]{Sample Worst-Case Discovery Times}
\acrodef{DTN}{Delay Tolerant Network} \acrodefplural{DTN}[DTN's]{Delay Tolerant Networks}
\acrodef{GCD}{Greatest Common Divisor} 
\acrodef{IoT}{Internet of Things}
\acrodef{ILP}{Integer Linear Program} \acrodefplural{ILP}[ILP's]{Integer Linear Programs}
\acrodef{LCM}{Least Common Multiple} \acrodefplural{LCM}[LCM's]{Least Common Multiples}
\acrodef{LP}{Linear Program} \acrodefplural{LP}[LP's]{Linear Programs}
\acrodef{MAC}{Media Access Control}
\acrodef{CR}{Cognitive Radio}
\acrodef{PU}{Primary User} \acrodefplural{PU}[PU's]{Primary Users}
\acrodef{SU}{Secondary User}\acrodefplural{SU}[SU's]{Secondary Users}
\acrodef{CH}{Channel Hopping}
\acrodef{ISM}{Industrial, Scientific and Medical}
\acrodef{MDT}{Mean Discovery Time} \acrodefplural{MDT}[MDT's]{Mean Discovery Times}
\acrodef{NDoT}{Number of Discoveries over Time}
\acrodef{OFDM}{Orthogonal Frequency-Division Multiplexing}
\acrodef{PSV}{Passive Scan}
\acrodef{SMDT}{Sample Mean Discovery Time} \acrodefplural{SMDT}[SMDT's]{Sample Mean Discovery Times}
\acrodef{SNDoT}{Sample Number of Discoveries over Time}
\acrodef{TU}{Time Unit} \acrodefplural{TU}[TU's]{Time Units}
\acrodef{CDF}{Cumulative Distribution Function} \acrodefplural{CDF}[CDF's]{Cumulative Distribution Functions}
\acrodef{PDF}{Probability Distribution Function} \acrodefplural{PDF}[PDF's]{Probability Distribution Functions}
\acrodef{WLAN}{Wireless Local Area Network} \acrodefplural{WLAN}[WLAN's]{Wireless Local Area Networks} 
\bibliography{references}

\end{document}